\documentclass[reqno]{amsart}

\usepackage[applemac]{inputenc}

\usepackage{amssymb}
\usepackage{graphicx}
\usepackage[cmtip,all]{xy}
\usepackage{verbatim}
\usepackage{tikz-cd}
\usepackage{mathtools}
\usepackage{graphicx}
\usepackage{mathrsfs}
\usepackage{geometry}
\usepackage{enumitem}
\usepackage{booktabs}

\DeclareMathAlphabet{\mathpzc}{OT1}{pzc}{m}{it}

\newtheorem{theorem}{Theorem}[section]
\newtheorem{lemma}[theorem]{Lemma}
\newtheorem{corollary}{Corollary}[section]

\theoremstyle{definition}
\newtheorem{definition}[theorem]{Definition}
\newtheorem{example}[theorem]{Example}

\theoremstyle{remark}
\newtheorem{remark}[theorem]{Remark}

\theoremstyle{question}

\numberwithin{equation}{section}

\makeatletter
\DeclareRobustCommand{\cev}[1]{%
  \mathpalette\do@cev{#1}%
}
\newcommand{\do@cev}[2]{%
  \fix@cev{#1}{+}%
  \reflectbox{$\m@th#1\vec{\reflectbox{$\fix@cev{#1}{-}\m@th#1#2\fix@cev{#1}{+}$}}$}%
  \fix@cev{#1}{-}%
}
\newcommand{\fix@cev}[2]{%
  \ifx#1\displaystyle
    \mkern#23mu
  \else
    \ifx#1\textstyle
      \mkern#23mu
    \else
      \ifx#1\scriptstyle
        \mkern#22mu
      \else
        \mkern#22mu
      \fi
    \fi
  \fi
}

\makeatother

 \newcommand{\virgolette}{``}

\newcommand{\cat}[1]{\ensuremath{\mbox{\sffamily{#1}}}}

\newcommand{\slantone}[2]{{\raisebox{.1em}{$#1$}\left/\raisebox{-.1em}{$#2$}\right.}}
\newcommand*{\defeq}{\mathrel{\vcenter{\baselineskip0.5ex \lineskiplimit0pt
                     \hbox{\scriptsize.}\hbox{\scriptsize.}}}%
                     =}

\newcommand\asim{\mathrel{%
  \ooalign{\raise0.1ex\hbox{$\sim$}\cr\hidewidth\raise-0.8ex\hbox{\scalebox{0.9}{$\scriptstyle{x}$}}\hidewidth\cr}}}

\newcommand{\proj}[1]{\ensuremath{\mathbb{P}^{#1}}}

\newcommand{\mani}{\ensuremath{\mathpzc{X}}}

\newcommand{\stsheaf}{\ensuremath{\mathcal{O}_{\mathpzc{M}}}}

\newcommand{\beq}{\begin{equation}}
\newcommand{\eeq}{\end{equation}}
\newcommand{\bear}{\begin{eqnarray}}
\newcommand{\eear}{\end{eqnarray}}

\begin{document}

\title{On BV Supermanifolds and the Super Atiyah Class}


%    author one information
\author{Simone Noja}
\address{Universit\"{a}t Heidelberg}
\curraddr{Im Neuenheimer Feld 205, 69120 Heidelberg, Germany}
\email{noja@mathi.uni-heidelberg.de}
%\thanks{}

%    author two information
%\author{Riccardo Re}
%\address{Università degli Studi dell'Insubria}
%\curraddr{Via Valleggio 11, 22100 Como, Italy}
%\email{riccardo.re@uninsubria.it}
%\thanks{}

%\subjclass[2000]{Primary }
%    The 2010 edition of the Mathematics Subject Classification is
%    now available.  If you are citing a classification from the
%    new scheme, use the following input coding instead.
%\subjclass[2010]{Primary }

%\date{}

\begin{abstract} We study global and local geometry of forms on odd symplectic BV supermanifolds, constructed from the total space of the bundle of 1-forms on a base supermanifold. 
We show that globally 1-forms are an extension of vector bundles defined on the base supermanifold. In the holomorphic category, we prove that this extension is split if and only if the super Atiyah class of the base supermanifold vanishes. This is equivalent to the existence of a holomorphic superconnection: we show how this condition is related to the characteristic non-split geometry of complex supermanifolds. From a local point of view, we prove that the deformed de Rham double complex naturally arises as a de-quantization of the de Rham/Spencer double complex of the base supermanifold. Following \v{S}evera, we show that the associated spectral sequence yields semidensities on the BV supermanifold, together with their differential in the form of a super BV Laplacian.  

%to the geometry of the reduced space and the obstructions to splitting the the base supermanifold

\end{abstract}

\maketitle

\tableofcontents

\section{Introduction}

%\noindent Originally designed as a tool for perturbative quantization of gauge theories, Batalin-Vilkovisky formalism (henceforth \virgolette BV formalism'') is now established as one of the prominent   \\

\noindent The Batalin-Vilkovisky formalism (henceforth the \virgolette BV formalism'') was originally designed in the early 1980's as a tool to deal with the perturbative quantization of gauge theories. Nowadays, its importance goes far beyond its original purpose: the BV formalism has grown into one of the foundational language of contemporary theoretical and mathematical physics \cite{Costello, Getzler, Mnev}, with several applications also to pure mathematics \cite{Sullivan, Pantev}. % On the foundational side, as testified for example by the recent Costello-Gwilliam' construction of local quantum field theories - where BV formalism plays a pivotal role {\bf Costello-Gwilliam}. 
%As such, BV formalism appears naturally in very different contexts. From a foundational point of view, it is a building pillar of the recent influential Costello-Gwilliam' construction of local quantum field theories via factorization algebras \cite{Costello}, but it also enters in several applications which span from general relativity \cite{Mitch1, Mitch2} to string field theory \cite{Zwiebach}, from conformal field theories \cite{Getzler} to holomorphic and \virgolette twisted'' supersymmetric field theories \cite{Ingmar1, Ingmar2}. The advent of BV formalism marked the entrance of homological algebra in theoretical physics: under its lens, the path integral looses its stigma of ill-defined analytical machinery, to become instead an algebraic or - better - a homological machinery, which is capable to account for different species of topological and geometric invariants related to the original physical theory. Indeed, BV formalism has proved itself very effective also in pure mathematics. Beside the celebrated Kontsevich's deformation quantization of Poisson manifolds, one should mention applications to algebraic topology - namely invariants of 3-manifolds, knots and string topology \cite{Sullivan} - and it has recently prompted advances in derived symplectic geometry \cite{Pantev}. \\
\noindent It was Albert Schwarz in \cite{Schwarz} who first elucidated the geometric framework that lies at the basis of the BV formalism, by recognizing the crucial role played by supergeometry. The BV formalism builds upon the BRST formalism, that in turns introduced a new point of view on the so-called Faddeev-Popov procedure. In particular, the BRST formalism identifies the space of fields of a gauge theory with \virgolette functions'' on a supermanifold $( \cat{M}_{\mathpzc{BRST}}, \mathcal{O}_{\tiny{\cat{M}}_{\mathpzc{BRST}}})$, typically constructed as the split supermanifold generated by an action Lie algebroid $ \mathfrak{g}\times M \rightarrow M$ associated to a Lie group action $G \times M \rightarrow M$, so that $\mathcal{O}_{\tiny{\cat{M}}_{\mathpzc{BRST}}}$ is trivialized as $\mathcal{O}_{\tiny{\cat{M}}_{\mathpzc{BRST}}} (U) \cong \mathcal{C}^{\infty} (U) \otimes \wedge^\bullet \Pi \mathfrak{g}^\ast$. %The supermanifold $\cat{M}_{\mathpzc{BRST}}$ is  endowed with a natural odd vector superfield $Q_{\mathpzc{BRST}}$, encoding the symmetries of the physical theory, \emph{i.e.}\ $Q_{\mathpzc{BRST}} (\mathcal{S}) = 0$, for $\mathcal{S}$ the $G$-invariant action functional, which thus defines a $\mathcal{Q}_{\mathpzc{BRST}}$-cocycle: these supermanifolds are indeed dubbed $Q$-(super)manifolds after the characterizing odd vector field $Q_{\mathpzc{BRST}}$. 
In this framework, ghost fields are not just byproducts of an integration procedure, as in the Faddeev-Popov procedure. Instead, they are conceptually understood as proper geometric data: odd sections of $\mathcal{O}_{\tiny{\cat{M}}_{\mathpzc{BRST}}}$ related to the generators of the Chevalley-Eilenberg algebra $\wedge^\bullet \Pi \mathfrak{g}^\ast$ resolving the (infinitesimal) gauge symmetry $\mathfrak{g}.$ The BV formalism enhances the above \virgolette BRST package'' by crucially allowing symplectic geometry to enter the picture: starting from $\cat{M}_{\mathpzc{BRST}}$, one constructs a canonically associated (odd) symplectic supermanifold $( \cat{M}_{\mathpzc{BV}}, \mathcal{O}_{\tiny{\cat{M}}_{\mathpzc{BV}}} )$ as the total space supermanifold of the (parity-shifted) cotangent bundle $\cat{M}_{\mathpzc{BV}} \defeq (\Pi \cat{T}^{\ast} \cat{M}_{\mathpzc{BRST}} \rightarrow \cat{M}_{\mathpzc{BRST}})$. This is called BV space or BV supermanifold. Functions on $\cat{M}_{\mathpzc{BV}}$ are given by polynomial functions on the fibers of $\cat{M}_{\mathpzc{BV}}$, that is $\mathcal{O}_{\tiny{\cat{M}}_{\mathpzc{BV}}} \defeq (\Omega^\bullet_{\tiny{\cat{M}}_\mathpzc{BRST}})^\ast$. This means that if in the BRST setting the local $n|m$ coordinates $(x^i, \theta^\alpha)$ of $\cat{M}_{\mathpzc{BRST}}$ are identified with the fields $x^i$ and the ghosts $\theta^\alpha$ of the related physical theory, in the BV setting these get supplemented by another set of $m|n$ coordinates $( q_\alpha | p_i)$, accounting for the fiber directions of $\cat{M}_{\mathpzc{BV}}$ and identified with the so-called anti-fields $p_i$ and anti-ghosts $q_\alpha$. In a similar fashion as above, whereas the odd ghosts $\theta_\alpha$ provided a homological resolution for the gauge symmetry, the new odd generators in $\cat{M}_{\mathpzc{BV}}$ - the anti-fields $q_i$'s - provide a homological (Koszul-Tate) resolution of the critical locus of the action $\mathcal{S}$, showing once again the vicinity of the formalism with a homological or, better, {derived} geometric point of view.\\

\begin{comment}

\vspace{.3cm}
\begin{center}
\begin{tabular}{cccccc}
\toprule
$\cat{M}_{\mathpzc{BRST}} $ & & $\mathpzc{even}$  & & $\mathpzc{odd}$ \\
\midrule
$\mathpzc{coordinates}$  &  & $x^i \leftrightsquigarrow \mathpzc{fields}$ & & $\theta^\alpha \leftrightsquigarrow  \mathpzc{ghosts}$ \\
%odd &  & $x^i$ = fields & $\theta^\alpha$  = ghosts & $p_i$ = anti-fields & $q_\alpha$ = anti-ghosts \\
\bottomrule
\end{tabular}
\end{center}
\vspace{.2cm}
\begin{center}
\begin{tabular}{cccccc}
\toprule
$\cat{M}_{\mathpzc{BV}} $ & & $\mathpzc{even}$  & & $\mathpzc{odd}$ \\
\midrule
$\mathpzc{base}$ $\mathpzc{coordinates}$  &  & $x^i \leftrightsquigarrow \mathpzc{fields}$ & & $\theta^\alpha \leftrightsquigarrow  \mathpzc{ghosts}$ \\
$\mathpzc{fiber}$ $\mathpzc{coordinates}$  &  & $q_\alpha \leftrightsquigarrow \mathpzc{anti}$-$\mathpzc{ghosts}$ & & $p_i \leftrightsquigarrow  \mathpzc{anti}$-$\mathpzc{fields}$ \\
\bottomrule
\end{tabular}
\end{center}
\vspace{.3cm}

\end{comment}

\noindent Generally speaking, the most interesting aspects of supergeometry are those which do not arise as a generalization of the ordinary commutative theory, but instead force us to revise our classical geometric intuition and to confront ourself with unexpected new features. These new features are both of local and global nature. Locally, the geometry of forms on supermanifolds and the related integration theory present the most peculiar non-trivial novelties: the failure of a trivial generalization of Poincaré duality leads to the introduction of a new kind of forms, which are crucial for the purpose of a meaningful integration theory and for these reason are called integral forms. On the other hand, globally, complex supermanifolds can be non-split or non-projected: this means that they cannot be reconstructed from ordinary geometric data, but instead they are genuinely new geometric spaces living a life of their own. \\
\noindent In this paper we relate these two aspects, by starting from the geometry of forms arising from a BV supermanifold of the type of $\cat{M}_{\mathpzc{BV}}$ above. %\emph{i.e.}\ arising as the total space of the parity shifted cotangent bundle of a certain supermanifold in any geometric category. \\%not necessarily of the kind of $\cat{M}_{\mathpzc{BRST}}$ above. \\
More precisely, the paper is organized as follows. In Section 2 we recall the main definitions of the supergeometric objects and constructions that we will use. We then address the global aspects related to the geometry of forms on BV supermanifolds in Section 3. In particular, we prove that 1-forms on $\cat{M}_{\mathpzc{BV}}$-type supermanifolds are an extension of vector bundles defined on the base supermanifold, Theorem \ref{extteo}. For real supermanifolds we show in Theorem \ref{splittingsmooth} that this extension is always split and the splitting corresponds to a reduction of the structural symplectic supergroup. More interestingly, building upon a projection to (the cohomology of) the base supermanifold of $\cat{M}_{\mathpzc{BV}}$, Theorem \ref{projection}, in the case of complex supermanifolds we show that the extension is split if and only if the super Atiyah class of the base supermanifold vanishes, corresponding to the existence of a holomorphic connection, Theorem \ref{PseudoAt}. We then show how this condition is related to the characteristic non-split or non-projected geometry of complex and algebraic supermanifolds in Theorem \ref{obsthm} and we comment on future perspectives in this context. Several examples spanning different levels of sophistication are discussed in Section 5. From a local point of view, in Section 6, we focus on the geometry related to the symplectic nature of the supermanifold $\cat{M}_{\mathpzc{BV}} $ and we show that the associated deformed de Rham double complex - whose differentials are given by the odd symplectic form and the de Rham differential - naturally arises as a de-quantization of the de Rham / Spencer double complex associated to the base supermanifold. Following \v{S}evera, we show in Theorem \ref{BVop} that the related spectral sequence yields semidensities on the odd symplectic BV supermanifold and their differential in the form of a super BV Laplacian.

\vspace{.3cm}

\noindent {\bf Acknowledgments.} The author wish to thank J. Huerta and J. Walcher for fruitful remarks and R. Re for conversations and collaboration on related projects. This work is funded by Deutsche Forschungsgemeinschaft (DFG, German Research Foundation) under Germany’s Excellence Strategy EXC-2181/1 - 390900948 (the Heidelberg STRUCTURES Cluster of Excellence).

\section{Main Definitions: Local and Global Data}

\noindent In this section we recollect the definitions of the main geometric objects that will be used throughout the paper. For thorough introductions to the geometry of supermanifolds the reader is suggest to refer to the books \cite{BR, Manin}. The recent \cite{Noja} offers a detailed account of the geometry of forms on supermanifolds. \\
We let $\mani$ be a {smooth}, {analytic} or {algebraic} supermanifold of dimension $p|q$ with structure sheaf given by $\mathcal{O}_\mani$ and we denote $\mani_{\mathpzc{red}}$ its \emph{reduced space}, which is an ordinary (commutative smooth, analytic or algebraic) manifold of dimension $p$ with structure sheaf $\mathcal{O}_{\mani_{\mathpzc{red}}}$. We will denote with $\mathcal{O}_{\mani, 0}$ and $\mathcal{O}_{\mani, 1}$ the even, respectively odd part of the structure sheaf with respect to its $\mathbb{Z}_2$-gradation, and likewise for other sheaves or vector bundles introduced in the following. We define $\cat{T}_\mani$ be the tangent sheaf of $\mani$. This is a locally-free sheaf of (left) $\mathcal{O}_\mani$-modules of rank $p|q$: if we let $U$ be an open set in the topological space $|\mani_{\mathpzc{red}}|$ underlying $\mani$ and $x_a \defeq z_i | \theta_\alpha$ for $i = 1, \ldots, p$ and $\alpha = 1, \ldots, q$ be a system of local coordinates over $U$ for $\mani,$ then
\bear
\cat{T}_\mani (U) = \mathcal{O}_{\mani} (U) \cdot \{\partial_{z_1}, \ldots, \partial_{z_p} | \partial_{\theta_1}, \ldots, \partial_{\theta_q} \}, 
\eear
where $\mathcal{O}_\mani$ is the structure sheaf of $\mani$ and the local generators $\partial_{z_i}$'s are even and the $\partial_{\theta_\alpha}$'s are odd. Given the tangent sheaf as defined above we can immediately introduce two related sheaves. The first one, is the cotangent sheaf $\cat{T}^\ast_{\mani}$, which is the \emph{dual} of $\cat{T}_\mani$, \emph{i.e.}\ $\cat{T}^\ast_\mani \defeq \mathcal{H}om_{\mathcal{O}_{\mani}} (\cat{T}_\mani, \mathcal{O}_\mani)$. The second one is the \emph{parity shifted} tangent sheaf $\Pi \cat{T}_\mani$, which is a locally-free sheaf of $\mathcal{O}_\mani$-module of rank $q|p$. With reference to the above trivialization over $U$, the parity shifted tangent sheaf is locally generated as follows:
\bear
\Pi \cat{T}_\mani (U) = \mathcal{O}_{\mani} (U) \cdot \{\pi \partial_{\theta_1}, \ldots, \pi \partial_{\theta_q} | \pi \partial_{z_1}, \ldots, \pi \partial_{z_p} \}, 
\eear
where we stress that the local generators $\pi \partial_{\theta}$'s are even and the $\pi \partial_{z}$'s are odd. Sections of the parity-shifted tangent sheaf are called $\Pi$-vector fields or vector fields for short. We call the \emph{dual} of the parity-shifted tangent sheaf the \emph{sheaf of $1$-forms} on $\mani$ and we denote it as $\Omega^1_{\mani} \defeq \Pi \cat{T}^\ast_\mani$. This is a locally-free sheaf of (right) $\mathcal{O}_\mani$-modules of rank $q|p$: over an open set $U$ one has the trivialization 
\bear
\Omega^1_{\mani} (U) = \{d\theta_1, \ldots, d\theta_q | dz_1, \ldots, dz_p \} \cdot \mathcal{O}_\mani (U),
\eear 
where the local generators $d\theta$'s are even and the $dz$'s are odd. The duality paring between $\Omega^1_\mani$ and $\cat{T}_\mani$ over $U$ reads $dx_a (\pi \partial_{x_b}) = \delta_{ab}$ for any $a,b$ both even or odd. The name sheaf of $1$-forms is in some sense justified by the fact that in supergeometry it is customary take the de Rham differential to be an odd morphism, so that one indeed has a sheaf morphism $d : \mathcal{O}_\mani \rightarrow \Omega^1_\mani$ which satisfies the $\mathbb{Z}_2$-graded Leibniz rule. Application of the (super)symmetric power functor $\cat{S}^\bullet : \mathbf{Sh}^{\mathbf{Mod}}_{\mathcal{O}_{\mani}} \rightarrow \mathbf{Sh}^{\mathbf{Alg}}_{\mathcal{O}_{\mani}}$ to the sheaf of 1-forms $\Omega^1_{\mani}$ and to $\Pi \cat{T}_{\mani}$ yields respectively the algebra of forms and $\Pi$-vector fields on $\mani.$ A section of $\Omega^k_\mani \defeq \cat{S}_{\mathcal{O}_{\mani}}^k \Omega^1_{\mani}$ is called a $k$-form and a section of $\cat{S}^k \Pi\cat{T}_\mani$ is called a ($k$-)$\Pi$-polyfield or $k$-polyfield for short. In this context the de Rham differential lifts to the \emph{exterior derivative} $d: \Omega^k_{\mani} \rightarrow \Omega^{k+1}_\mani$, which is an odd derivation of $\Omega^\bullet_\mani$, \emph{i.e.} it obeys the $\mathbb{Z}_2$-graded Leibniz rule in the form
\bear
d (\omega \eta ) = d \omega  \eta + (-1)^{|\omega|} \omega d \eta, 
\eear 
where $\omega$ and $\eta$ are two generic forms in $\Omega^\bullet_\mani$ and where we have denoted $|\omega|$ the $\mathbb{Z}_2$-degree, henceforth \emph{parity} of $\omega$. Further, it is easy to see that the exterior derivative is nilpotent: the pair $(\Omega_\mani^\bullet, d)$ defines a sheaf of $dg$-algebras, the \emph{de Rham complex} of $\mani$. \\
As it is well-known, differential forms are not suitable for integration in a supergeometric setting \cite{Manin}: this leads to the introduction of a second complex, which is \virgolette dual'' to the de Rham complex. This is the so-called complex of \emph{integral forms}, which are defined as sections of the tensor product of sheaves $\mathcal{B}er (\mani) \otimes_{\mathcal{O}_\mani} (\Omega^\bullet_\mani)^\ast,$ where $\mathcal{B}er (\mani) \defeq \mathcal{B}er (\Omega^1_\mani)^\ast$ is the Berezinian sheaf of $\mani$, see \cite{Manin}, \cite{NojaRe}, \cite{Noja}, which substitutes the notion of canonical sheaf of an ordinary manifold and whose sections are tensor densities. The interested reader can refer to \cite{Manin} or the recent \cite{CNR} and \cite{Noja} for a construction \emph{ab initio} of the differential for integral forms.\\
After this preliminary recollections of definitions, conventions and notations we introduce one of the main object under study in this paper.
\begin{definition}[The BV Supermanifold $\cat{M}$] \label{defM} Let $\mani$ be a smooth, analytic or algebraic supermanifold of dimension $p|q$ and let $\Omega^1_{\mani}$ be its sheaf of $1$-forms. We call $\cat{M} \defeq \cat{Tot} (\Omega^1_\mani \stackrel{\pi}{\longrightarrow} \mani)$ the $p+q | p + q$-dimensional supermanifold defined as a ringed space by the pair $(|\cat{M}_{\mathpzc{red}}|, \mathcal{O}_{\tiny{\cat{M}}})$, where the topological space $|\cat{M}_{\mathpzc{red}}|$ is given by the total space $|\cat{M}_{\mathpzc{red}}| \defeq \cat{Tot} ( (\Pi \cat{T}^\ast_{\mani})_0 \stackrel{\tilde \pi}{\longrightarrow} \mani_{\mathpzc{red}})$ of the vector bundle $(\Pi \cat{T}^\ast_\mani)_0 \stackrel{\tilde\pi}{\rightarrow} \mani_{\mathpzc{red}}$ of rank $p+q$ endowed with its canonical topology, %$|\cat{M}_{\mathpzc{red}}| \defeq \bigsqcup_{x \in \mani} (\Pi \cat{T}^\ast_{\mani, x})_0$, \emph{i.e.}\ the total space of even part of the sheaf of 1-forms on $\mani$, 
and the structure sheaf $\mathcal{O}_{\tiny{\cat{M}}}$ is defined as $\mathcal{O}_{\tiny{\cat{M}}} \defeq (\Omega^\bullet_\mani)^\ast$, where $(\Omega^\bullet_{\mani})^\ast$ is taken with its $\mathcal{O}_{\mani}$-module structure.% \emph{i.e.}\ functions on $\cat{M}$ are polynomial functions on the fibers $\cat{M}_x = \Pi \cat{T}^\ast_{\mani, x}.$
\end{definition}
\begin{remark} Notice that the canonical topology on the total space of a vector bundle is defined locally via the product topology, and then glueing along the transition functions: the related quotient topology is the desired topology on the total space of the vector bundle.  
\end{remark}
\begin{remark} The previous definition says that \virgolette functions'' on $\cat{M}$ are \emph{polynomial functions} on the fibers $\cat{M}_x = \Pi \cat{T}^\ast_{\mani, x}$, \emph{i.e.} they are polyfields having shifted parity. Locally, on open sets of the kind $\pi^{-1} (U) \cong U \times \mathbb{K}^{p+q}$ for $U$ an open set in $|\mani_{\mathpzc{red}}|$, one has that 
\bear
\mathcal{O}_{\cat{\tiny{M}}} (\pi^{-1} (U)) \defeq \mathcal{O}_\mani (U) \otimes_{\mathbb{K}}\mathbb{K}[\mathpzc{F}_1, \ldots \mathpzc{F}_{p+q}],
\eear  
for even and odd fiber coordinates $\{\mathpzc{F}_i \}_{i = 1, \ldots, p+q}$ and $\mathbb{K}$ the field of real or complex numbers. In the following subsection we will give an explicit local description via transition functions.  
\end{remark}
\begin{remark} A notational remark is in order. Indeed, in the introduction of the paper we have denoted the supermanifold $\cat{M}$ defined above in \ref{defM} with $\cat{M}_{\mathpzc{BV}}$, to distinguish it from the supermanifold $\cat{M}_{\mathpzc{BRST}},$ arising in the context of the BRST formalism. In the following we will always consider the supermanifold $\cat{M} = \cat{M}_{\mathpzc{BV}}.$
\end{remark}

%\noindent {\bf Domande: \\
%Definizione dello spazio topologico... prendo quello o $\Omega^1_\mani$?
%Forse bisogna aggiungere dipendenza da \mani... es: $\cat{M}_{\mani}$
%}

\subsection*{Local Description} It is worth to provide a {local} description of $\cat{M}$ in terms of transition functions of its local coordinates. We let $(U, x_a)$ be a local chart for the $p|q$-dimensional supermanifold \mani, where we stress that the index $a$ spans both even and odd local coordinates. Then $(\pi^{-1} (U), x_a, p_a)$ is a chart for $\cat{M}$ with 
\bear
p_a \defeq (-1)^{|x_a|} \partial_{dx_a}.
\eear
The following is immediate.
\begin{lemma}[Transition Functions of $\cat{M}$] \label{transM} Let $(U, x_a)$ and $(V, z_b)$ be two charts on $\mani$ with $U \cap V \neq \emptyset$, and let $(\pi^{-1}(U), x_a, p_a)$ and $(\pi^{-1} (V), z_b, q_b)$ the corresponding open sets on $\emph{\cat{M}}$. Then the transition functions of $\cat{\emph{M}}$ read
\bear
x_a = z_a (x), \qquad p_a = (-1)^{|x_a| + |z_b|} \left ( \frac{\partial z_b}{\partial x_a} \right ) q_b.
\eear
\end{lemma}
\begin{proof} The first ones are obvious, being the transition functions on $\mani$. For the latters, it is enough to observe that from $d z_b = dx_a (\partial_{x_a} z_b)$ it follows that
\begin{align}
p_a & = (-1)^{|x_a|} \partial_{dx_a} =  (-1)^{|x_a|} \partial_{dx_a} \left ( dx_c \frac{ \partial z_b}{\partial x_c} \right ) \partial_{dz_b} = (-1)^{|x_a| + |z_b|}  \left ( \frac{\partial z_b}{\partial x_a} \right ) q_b,
\end{align}
where we have made use of the definition of $q_b$ in the last step.
\end{proof}

\section{The Geometry of Forms: Split and Non-Split Extensions}

\noindent We now study the geometry of the cotangent sheaf $\Omega^1_{\tiny{\cat{M}}}$ of the supermanifold $\cat{M}$. Note that this is a locally-free sheaf of $\mathcal{O}_{\tiny{\cat{M}}}$-modules of rank  $p+q | p+q$. We can characterize its transition functions thanks to Lemma \ref{transM}.
\begin{lemma}[Transition Functions of $\Omega^1_{\tiny{\cat{M}}}$] \label{transCotM} Let $\emph{\cat{M}}$ be defined as above and let $(dx_a, dp_a)$ and $(dz_b, dq_b)$ be two local bases of $\Omega^1_{\tiny{\cat{\emph{M}}}}$ on the open sets $\pi^{-1} (U)$ and $\pi^{-1}(V)$ on $\emph{\cat{M}}$ with $U \cap V \neq \emptyset $. Then the transition functions of $\Omega^1_{\tiny{\cat{\emph{M}}}}$ read
\begin{align} \label{tfcot1}
dx_a = dz_b \left ( \frac{\partial x_a}{\partial z_b} \right ), 
\end{align}
\begin{align} \label{tfcot2}
&dp_a = \left (\frac{\partial z_b}{\partial x_a } \right ) dq_b + (-1)^{|x_a| + |z_b|} d \left ( \frac{\partial z_b}{\partial x_a} \right ) q_b. 
\end{align}
\end{lemma}
\begin{proof} The first ones are obvious. For the transition functions of the $dp$'s we observe that we have
\bear
d p_a = dz_b \left ( \frac{\partial p_a}{\partial z_b} \right ) + dq_b \left ( \frac{\partial p_a}{\partial q_b} \right ). 
\eear
The first summand reads
\begin{align}
dz_b \frac{\partial p_a }{\partial z_b} & = dz_b \frac{\partial}{\partial z_b} \left ( (-1)^{|x_a| + |z_b|} \frac{\partial z_c}{\partial x_a} q_c \right ) = (-1)^{|x_a| + |z_b|} d \left ( \frac{\partial z_b }{\partial x_a } \right) q_b.
\end{align}
The second summands reads 
\begin{align}
dq_b \left ( \frac{\partial p_a}{\partial q_b} \right ) & =  dq_b \frac{\partial}{\partial q_b} \left ( (-1)^{|x_a| + |z_c|} \left ( \frac{\partial z_c}{\partial x_a }\right ) q_c \right )  = \frac{\partial z_b}{\partial x_a} dq_b.
%& = (-1)^{|x_a| + |x_b| + (|z_b| + 1)(|z_b| + |y_a|) + |z_b|(|z_b| + |x_a|)} \frac{\partial z_b}{ \partial x_a} dq_b \nonumber \\ 
\end{align}
%thus concluding the proof.
\end{proof}
\noindent The previous lemma describes $\Omega^1_{\tiny{\cat{M}}}$ locally in terms of its transition functions, but it yields informations also on its global geometry, as the following shows. 
\begin{theorem}[$\Omega^1_{\tiny{\cat{M}}}$ as Extension of Vector Bundles] \label{extteo}Let $\emph{\cat{M}}$ be defined as above. Then the canonical exact sequence 
\bear
\xymatrix{
0 \ar[r] & \pi^{\ast} \Omega^1_{\mani} \ar[r] & \Omega^{1}_{\tiny{{\cat{\emph{M}}}}} \ar[r] & \Omega^1_{\tiny{\cat{\emph{M}}} / \mani} \ar[r] & 0,
}
\eear
induces the isomorphism $\Omega^1_{\tiny{\cat{\emph{M}}} / \mani} \cong \pi^\ast \emph{\cat{T}}_\mani$. In particular, $\Omega^1_{\tiny{\cat{\emph{M}}}}$ is an extension of locally-free sheaves 
\bear \label{extomega}
\xymatrix{
0 \ar[r] & \pi^{\ast} \Omega^1_{\mani} \ar[r] & \Omega^{1}_{\tiny{{\cat{\emph{M}}}}} \ar[r] & \pi^\ast \emph{\cat{T}}_\mani \ar[r] & 0.
}
\eear
\end{theorem}
\begin{proof} We work in the same setting of Lemma \ref{transCotM}. We first observe that the transformations of equation \eqref{tfcot1} identify the sections $(dx_a)$'s as a local basis of $\pi^\ast \Omega^1_{\mani}$ (notice the slight abuse of notation). The first summand in the transformations given by equation \eqref{tfcot2} identify the transformations of the parity-reversed dual of $\pi^\ast \Omega^1_{\mani}$, as the $dp_a$ have opposite parity with respect to the $dx_a$. This is hence identified with $\pi^\ast \cat{T}_\mani.$ The second summand in \eqref{tfcot2} gives the off-diagonal term of the extension of $\pi^\ast \cat{T}_{\mani}$ with $\pi^\ast \Omega^1_\mani.$
%$ \{ \partial_{x_a} \}$ to the basis $\{\partial_{z_b} \}$ of $\pi^\ast \cat{T}_\mani$ in the open set $U$ and $V$ respectively, with $U \cap V \neq \emptyset.$
 \end{proof}
%\begin{definition}[$\Omega^1_{\tiny{\cat{M}}}$-extension]
%We call $\Omega^1_{\tiny{\cat{M}}}$\emph{-extension} the short exact sequence \eqref{extomega}. 
%\end{definition} 

\noindent It follows from the previous theorem that in order to study the geometry of $\Omega^1_{\tiny{\cat{M}}}$ one needs to consider the cohomology group 
\bear \label{extprimo}
Ext^1_{\mathcal{O}_{\cat{\tiny{{M}}}}} (\pi^\ast \cat{T}_{\mani}, \pi^\ast \Omega^1_{\mani}) \cong H^1 (|\cat{M}_{\mathpzc{red}}|, \mathcal{H}om_{\mathcal{O}_{\cat{\tiny{M}}}} (\pi^\ast \cat{T}_{\mani}, \pi^\ast \Omega^1_{\mani} )),
\eear
which controls the splitting of the exact sequence \eqref{extomega}, that will be called $\Omega^1_{\tiny\cat{M}}$\emph{-extension} in the rest of the paper. For ease of reading, we have deferred to the appendix a very concrete construction of the $Ext$-group related to an extension, which highlights the structure of the representatives in the above cohomology group in terms of the transition functions of the vector bundles involved. As we shall see, this concrete approach will play a significant role in what follows.\\

\noindent It is convenient to re-express this Ext-group appearing in \eqref{extprimo} as a cohomology group computed on the supermanifold $\mani$ - and hence on $\mani_{\mathpzc{red}}$ - instead of $\cat{M}$. 
\begin{theorem}[Projection to $\mani$] \label{projection} Let $\mani$ be a smooth, analytic or algebraic supermanifold, and let $\cat{\emph{M}} $ be constructed as above with $\pi : \cat{\emph{M}} \rightarrow \mani$ its projection map. Then one has the following natural isomorphism
\bear
Ext^1_{\mathcal{O}_{\cat{\tiny{\emph{M}}}}} (\pi^\ast \cat{\emph{T}}_{\mani} , \pi^\ast \Omega^1_{\mani}) \cong H^1 (|\mani_{\mathpzc{red}}|, \cat{\emph{T}}_{\mani}^\ast \otimes_{\mathcal{O}_{\mani}} \mathcal{E}nd_{\mathcal{O}_\mani} (\cat{\emph{T}}_{\mani} )).
\eear
\end{theorem}
\begin{proof} First, notice that
\bear
Ext^1_{\mathcal{O}_{\cat{\tiny{{M}}}}} (\pi^\ast \cat{T}_{\mani} , \pi^\ast \Omega^1_{\mani}) \cong H^1 (|\cat{M}_{\mathpzc{red}}|, \pi^\ast \mathcal{H}om_{\mathcal{O}_\mani} (\cat{{T}}_\mani, \Omega^1_{\mani})).
\eear
Since in the given hypotheses, $\pi : \cat{M} \rightarrow \mani$ is an affine morphism, then by Leray's spectral sequence
\bear
H^1 (|\cat{M}_{\mathpzc{red}}|, \pi^\ast \mathcal{H}om_{\mathcal{O}_\mani} (\cat{{T}}_\mani, \Omega^1_{\mani})) \cong H^1 (|\mani_{\mathpzc{red}}|, \pi_\ast \pi^\ast \mathcal{H}om_{\mathcal{O}_\mani} (\cat{{T}}_\mani, \Omega^1_{\mani})). 
\eear
Finally, by projection formula applied to $R^i \pi_{\ast} - $ in the case $i=0$ (see \cite{Har} page 253) we have 
\bear
H^1 (|\mani_{\mathpzc{red}}|, \pi_\ast \pi^\ast \mathcal{H}om_{\mathcal{O}_\mani} (\cat{{T}}_\mani, \Omega^1_{\mani})) \cong H^1 (|\mani_{\mathpzc{red}}|, \mathcal{H}om_{\mathcal{O}_\mani} (\cat{{T}}_\mani, \Omega^1_{\mani}) \otimes_{\mathcal{O}_\mani} \pi_\ast \mathcal{O}_{\tiny{\cat{M}}}).
\eear
Further, since $\pi_\ast \mathcal{O}_{\tiny{\cat{M}}} \cong (\Omega^\bullet_{\mani})^\ast$ as $\mathcal{O}_\mani$-modules, this can be rewritten as 
\bear
H^1 (|\mani_{\mathpzc{red}}|, \mathcal{H}om_{\mathcal{O}_\mani} (\cat{{T}}_\mani, \Omega^1_{\mani}) \otimes_{\mathcal{O}_\mani} \pi_\ast \mathcal{O}_{\tiny{\cat{M}}}) \cong H^1 (|\mani_{\mathpzc{red}}|, \mathcal{H}om_{\mathcal{O}_\mani} (\cat{{T}}_\mani, \Omega^1_{\mani}) \otimes_{\mathcal{O}_\mani}  (\Omega^\bullet_{\mani})^\ast).
\eear
Finally, the linear dependence on $p$ in the second summand of the \eqref{tfcot2} shows that this extension class appears in degree one only in $(\Omega^\bullet_\mani)^\ast$, \emph{i.e.}\ in the summand $H^1 (|\mani_{\mathpzc{red}}|, \mathcal{H}om_{\mathcal{O}_\mani} (\cat{{T}}_\mani, \Omega^1_{\mani}) \otimes_{\mathcal{O}_\mani} (\Omega^1_{\mani})^\ast)$ of the above direct image, so that one finds
\bear
Ext^1_{\mathcal{O}_{\cat{\tiny{M}}}} (\pi^\ast \cat{T}_{\mani} , \pi^\ast \Omega^1_{\mani}) \cong H^1 (|\mani_{\mathpzc{red}}|, \mathcal{H}om_{\mathcal{O}_\mani} (\cat{{T}}_\mani, \Omega^1_{\mani}) \otimes_{\mathcal{O}_\mani} \Pi \cat{T}_\mani). 
\eear
The conclusions follows observing that $\mathcal{H}om_{\mathcal{O}_\mani} (\cat{{T}}_\mani, \Omega^1_{\mani}) \otimes_{\mathcal{O}_\mani} \Pi \cat{T}_\mani \cong \mathcal{H}om_{\mathcal{O}_\mani} (\cat{{T}}_\mani, \Pi \cat{T}_\mani) \otimes_{\mathcal{O}_\mani}  \Omega^1_{\mani} $ and that $\Omega^1_{\mani} \defeq \Pi \cat{T}_\mani^\ast \cong \Pi \mathcal{O}_{\mani} \otimes_{\mathcal{O}_\mani} \cat{T}^\ast_{\mani}$, so that $\mathcal{H}om_{\mathcal{O}_\mani} (\cat{T}_\mani,  \Pi \mathcal{O}_\mani \otimes_{\mathcal{O}_\mani} \cat{T}_\mani )  \otimes_{\mathcal{O}_\mani} \Pi \mathcal{O}_\mani \otimes_{\mathcal{O}_\mani} \cat{T}^\ast_\mani \cong \mathcal{E}nd (\cat{T}_\mani) \otimes_{\mathcal{O}_\mani} \cat{T}^\ast_\mani.$
\end{proof}

\noindent The above theorem can be applied to smooth real supermanifolds, as to show the existence of a reduction of the structure group of $\Omega^{1}_{\tiny{{\cat{{M}}}}}$. %to the extension $\Omega^1_{\tiny{\cat{M}}}$ as in equation \eqref{extomega}, providing a possibily more explicit form of Theorem \ref{splitting}.
To this end, following \cite{Manin} (see chapter 4, section 10), we recall that the structure group of $\Omega^{1}_{\tiny{{\cat{{M}}}}}$ is given by the \emph{symplectic supergroup} $\Pi Sp (p+q|p+q)$, that can be understood as the stabilizer of the \virgolette metric'' in $H^0 (|\cat{M}_{\mathpzc{red}}|, \Omega^1_{\tiny{\cat{M}}} \otimes \Omega^1_{\tiny{\cat{M}}})$ given by the \emph{odd symplectic form} $\omega$ - whose related geometry will be discussed in Section 6, see Definition \ref{oddsympl}.

\begin{theorem}[Splitting \& Reduction of Symplectic Supergroup] \label{splittingsmooth} Let $\mani$ be a smooth supermanifold and let $\cat{\emph{M}} = \cat{\emph{Tot}} (\Omega^1_{\mani})$ be the smooth supermanifold associated to $\mani$ as defined above. Then the following are true.
\begin{enumerate}[leftmargin=*]
\item[(1)] The $\Omega^1_{\tiny{\cat{\emph{M}}}}$-extension 
\bear
\xymatrix{
0 \ar[r] & \pi^{\ast} \Omega^1_{\mani} \ar[r] & \Omega^{1}_{\tiny{{\cat{\emph{M}}}}} \ar[r] & \pi^\ast \emph{\cat{T}}_\mani  \ar[r] \ar@{->}@/_1.2pc/[l] & 0
}
\eear
is split, \emph{i.e.} $\Omega^{1}_{\tiny{{\cat{\emph{M}}}}} \cong \pi^{\ast} \Omega^1_{\mani} \oplus  \pi^\ast \emph{\cat{T}}_\mani$ non-canonically.
\item[(2)] There exists a reduction of the structure group of $\Omega^1_{\tiny{\cat{\emph{M}}}}$ as follows 
\bear
\Pi Sp (p+q|p+q) \longrightarrow \left \{ \left ( \begin{array}{c|c} 
T & \\
\hline 
& (T^{-1})^{st}
\end{array}
\right ):
T \in GL (p|q)
\right \},
\eear
where $\Pi Sp(p+q|p+q)$ is the symplectic supergroup.
\end{enumerate}
\end{theorem}
\begin{proof} For the first point it is enough to observe that the existence of a smooth partition of unity in the smooth category leads to the exactness of the \v{C}ech cochain complex of any sheaf in degree $i > 0$, which is therefore fine, thus soft and acyclic. Applying this to $\cat{{T}}_{\mani}^\ast \otimes_{\mathcal{O}_\mani} \mathcal{E}nd_{\mathcal{O}_\mani} (\cat{{T}}_{\mani} )$ yields the conclusion, \emph{i.e.} $H^1 (|\mani_{\mathpzc{red}}|, \cat{{T}}_{\mani}^\ast \otimes_{\mathcal{O}_\mani} \mathcal{E}nd_{\mathcal{O}_\mani} (\cat{{T}}_{\mani} )) = 0$.\\
The second point follows from the first one and Theorem \ref{reductionlemma} in the Appendix, which generalizes to the $\mathbb{Z}_2$-graded context. To this end it is enough to observe that the structure of the transition functions as in \eqref{transomega} follows from Lemma \ref{transCotM}. %this leads to transition functions of the claimed form.
\end{proof}
\begin{remark} It is to be noted that the above splitting is non-canonical. It would be interesting to see if it is possibile to classify or provide constraints on smooth supermanifolds such that the above reduction of the structure group is possible via a suitable choice of charts.
\end{remark}
%\begin{remark} In physical applications the splitting of the $\Omega^1_{\tiny{\cat{M}}}$-extension is often left understood, and forms are indeed considered as given by the direct sum of those coming from the base supermanifold and those coming from the fiber.
%\end{remark}

\section{Connections and Obstructions: the Super Atiyah Class }

\noindent Theorem \ref{splittingsmooth} proves the existence of a splitting for the $\Omega^1_{\tiny{\cat{M}}}$-extension in the smooth category: this fact could have been easily inferred directly from equation \eqref{extprimo}. On the other hand, the \virgolette projection'' result of Theorem \ref{projection} allows for a very nice interpretation when working in the complex analytic or algebraic category, where sheaves admit non-trivial higher cohomologies and the splitting of the $\Omega^1_{\tiny{\cat{M}}}$-extension is far from obvious. Quite the contrary, we will see that in general the conditions under which the $\Omega^1_{\tiny{\cat{M}}}$-extension splits are quite restrictive. To this end, in the following we restrict ourself to work on complex supermanifolds in the holomorphic category: the reader shall see that everything holds true also in the algebraic category.    
%\noindent The previous theorem opens to a very nice interpretation. To this end we start with the following definition {\bf citazioni}.
\begin{definition}[Affine Connection on $\mani$] Let $\mani$ be a complex supermanifold and let $\cat{T}_\mani$ be the (holomorphic) tangent sheaf of $\mani.$ An affine connection on $\mani$ is an even morphism of sheaves of $\mathbb{C}$-vector spaces $\nabla : \cat{T}_\mani \rightarrow \cat{T}^\ast_\mani \otimes_{\mathcal{O}_\mani} \cat{T}_\mani$ such that it satisfies the Leibniz rule  
\bear
\nabla (fX) = d_{\mathpzc{ev}} f \otimes X + f \nabla X,
\eear
for any $f \in \mathcal{O}_\mani$ and $X \in \cat{T}_\mani$, where $d_{\mathpzc{ev}}: \mathcal{O}_\mani \rightarrow \cat{T}^\ast_\mani$ is the \emph{even} de Rham differential, see \cite{Manin}. 
%\begin{enumerate}[leftmargin=*]
%\item it is $\mathcal{O}_\mani$-linear in its first argument, \emph{i.e.}\ for any $f \in \mathcal{O}_{\mani}$ and $X, Y \in \cat{T}_\mani$ \
%\bear
%\nabla (f X \otimes Y) = f \, \nabla (X \otimes Y); 
%\eear
%\item it satisfies the Leibniz rule in its second argument, \emph{i.e.}\ for any $f \in \mathcal{O}_\mani$ and $X, Y \in \cat{T}_\mani$ 
%\bear
%\nabla (X \otimes f Y ) = X (f) Y + (-1)^{|X| |f|} f \nabla (X \otimes Y).
%\eear
%\end{enumerate} 
\end{definition} 
%\begin{remark} The above is usually called affine \emph{holomorphic} connection. Notice that an affine connection can be analogously characterized as a sheaf morphism $\nabla : \cat{T}_\mani \rightarrow \mathcal{H}om (\cat{T}_\mani, \cat{T}_\mani ) = \cat{T}^\ast_{\mani} \otimes \cat{T}_\mani $ together with the corresponding Leibniz rule $\nabla (fX ) = d_{ev} f \otimes X + f\, \nabla X$, where $d_{ev} : \mathcal{O}_{\mani} \rightarrow \cat{T}_{\mani}^\ast$ is the \emph{even} de Rham differential, see \cite{Kos} \cite{Manin}. 
%\end{remark}
\noindent Obstructions to the existence of an affine connection on a complex supermanifold \cite{BR} \cite{Betta} \cite{DonWit} \cite{Kos} can be established in same fashion of the original Atiyah's result \cite{At} for ordinary complex manifolds. We spell out the main points of the construction following \cite{BR}, which is very close to the original \cite{At}. \\

\noindent First, one defines the \emph{sheaf of $1$-jets} of $\cat{T}_\mani$. One starts introducing the sheaf of $\mathbb{C}$-vector spaces given by 
\bear
U \longmapsto \mathcal{J}^1 \cat{T}_\mani (U) \defeq \cat{T}_\mani (U) \oplus (\cat{T}^\ast_{\mani} \otimes_{\mathcal{O}_\mani (U)} \cat{T}_\mani) (U).
\eear 
for $U$ an open set of $\mani.$ Notice that sections of $\cat{T}^\ast_{\mani} \otimes_{\mathcal{O}_\mani} \cat{T}_\mani$ are 1-forms valued in the tangent bundle. The sheaf $\mathcal{J}^1 \cat{T}_\mani$ can be endowed with the structure of sheaf of $\mathcal{O}_\mani$-modules as follows: let $\mathpzc{j} \defeq (X , \tau) \in \mathcal{J}^1 \mathcal{T}_\mani (U) $ and $f \in \mathcal{O}_\mani (U)$. One defines the product
\bear \label{OmodJet}
f \cdot \mathpzc{j} = f \cdot (X, \tau)  \defeq (f X , f \tau  +  d_{\mathpzc{ev}}f \otimes X),
\eear
where $d_{\mathpzc{ev}} : \mathcal{O}_\mani \rightarrow \cat{T}^\ast_\mani$ is the {even} de Rham differential, see again \cite{Manin}. One can then verify that the sequence of sheaves of $\mathcal{O}_\mani$-modules given by
\bear \label{1jets}
\xymatrix{
0 \ar[r] & \cat{T}^\ast_\mani \otimes_{\mathcal{O}_\mani} \cat{T}_\mani \ar[r]^\alpha & \mathcal{J}^1 \cat{T}_\mani \ar[r]^{\; \; \beta} & \cat{T}_\mani \ar[r] & 0,
}
\eear
where 
$
\alpha (\tau) \defeq (0, \tau),$ and $\beta ((X, \tau )) = X,$ is exact. Notice that since $\cat{T}_\mani$ is locally-free, then the sequence \eqref{1jets} is locally split, hence there exists a covering $\{U_a \}_{a \in I}$ such that $\cat{T}_\mani |_{U_a}$ and $\cat{T}^\ast_\mani \otimes \cat{T}_\mani |_{U_a}$ are free and 
\bear
\mathcal{J}^1 \cat{T}_\mani |_{U_a} \cong  \cat{T}_\mani |_{U_a} \oplus (\cat{T}^\ast_\mani \otimes \cat{T}_\mani) |_{U_a} \cong \mathcal{O}_\mani^{\oplus n|m}|_{U_a} \oplus \mathcal{O}_{\mani}^{n^2+ m^2 | 2nm} |_{U_a}, 
\eear
which guarantees that $\mathcal{J}^1 (\cat{T}_\mani)$ is locally-free. On the other hand, due to the non-trivial $\mathcal{O}_\mani$-module structure of the sheaf of 1-jets of $\cat{T}_\mani$, the previous short exact sequence of locally-free sheaves of $\mathcal{O}_\mani$-modules \eqref{1jets} - henceforth \emph{1-jets short exact sequence} - does not necessarily split. 
Applying the functor $\mathcal{H}om (\cat{T}_\mani, -) \defeq \mathcal{H}om_{\mathcal{O}_\mani} (\cat{T}_\mani, -)$, and taking the long exact sequence in cohomology one gets
\bear \label{1jetscoho}\hspace*{-1.5cm} 
\xymatrix{
0 \ar[r] & H^0 (|\mani_{\mathpzc{red}}|, \cat{T}^\ast_\mani \otimes \mathcal{H}om (\cat{T}_\mani, \cat{T}_\mani )  ) \ar[r] & H^0 (|\mani_{\mathpzc{red}}|, \mathcal{H}om (\cat{T}_\mani,\mathcal{J}^1\cat{T}_\mani)) \ar[r] & H^0 (|\mani_{\mathpzc{red}}|, \mathcal{H}om (\cat{T}_\mani, \cat{T}_\mani)) \ar@(dr,ul)[dll]_{\delta} & \\
& H^1 (|\mani_{\mathpzc{red}}|,  \cat{T}^\ast_\mani \otimes \mathcal{H}om (\cat{T}_\mani, \cat{T}_\mani) ) \ar[r] & H^1 (|\mani_{\mathpzc{red}}|, \mathcal{H}om (\cat{T}_\mani,\mathcal{J}^1\cat{T}_\mani)) \ar[r] & H^1 (|\mani_{\mathpzc{red}}|, \mathcal{H}om (\cat{T}_\mani, \cat{T}_\mani)) \ar[r] & \ldots
}
\eear
We call this long exact sequence in cohomology the \emph{1-jets long exact cohomology sequence}. %The connection homomorphism reads 
%\bear
%\delta : H^0 (|\mani_{\mathpzc{red}}|, \mathcal{H}om (\cat{T}_\mani, \cat{T}_\mani)) \rightarrow H^1 (|\mani_{\mathpzc{red}}|, \mathcal{H}om (\cat{T}_\mani, \cat{T}_\mani) \otimes \cat{T}^\ast_\mani).
%\eear 
We can thus give the following definition.
\begin{definition}[Super Atiyah Class] Let $\mani$ be a complex supermanifold and let $\cat{T}_\mani$ be its tangent sheaf. We define the Atiyah class $\mathfrak{At} (\cat{T}_\mani)$  of $\cat{T}_\mani$ to be the image of the identity map $id_{ \cat{\tiny{{T}}}_\mani} \in H^0 (|\mani_{\mathpzc{red}}|, \mathcal{H}om_{\mathcal{O}_{\mani}} (\cat{T}_\mani, \cat{T}_\mani))$ via the $1$-connecting homomorphism $\delta$ in the $1$-jets long exact cohomology sequence, \emph{i.e.}  %$ \delta   (id_{\cat{\tiny{{T}}}^\ast_\mani \otimes \cat{\tiny{{T}}}_\mani} ) \in H^1 (|\mani|, \mathcal{H}om (\cat{T}_\mani, \cat{T}_\mani) \otimes \cat{T}^\ast_\mani)$ of the identity $id_{\cat{\tiny{{T}}}^\ast_\mani \otimes \cat{\tiny{{T}}}_\mani} \in H^0 (|\mani|, \mathcal{H}om (\cat{T}_\mani, \cat{T}_\mani))$ via the $1$-connecting homomorphism $\delta$ in the $1$-jets long exact cohomology sequence the Atiyah class of $\cat{T}_\mani$, \emph{i.e.} 
\bear
\xymatrix@R=1.5pt{
\mathfrak{At} : H^0 (|\mani_{\mathpzc{red}}|, \mathcal{H}om_{\mathcal{O}_\mani} (\cat{{T}}_\mani, \cat{{T}}_\mani)) \ar[r] & H^1 (|\mani_{\mathpzc{red}}|,  \cat{{T}}^\ast_\mani \otimes \mathcal{H}om_{\mathcal{O}_\mani} (\cat{{T}}_\mani,\cat{{T}}_\mani ) ) \\
id_{ \cat{\tiny{{T}}}_\mani} \ar@{|->}[r] & \mathfrak{At} ({\cat{{{T}}}_\mani}) \defeq \delta   (id_{  \cat{\tiny{{T}}}_\mani} ). 
}
\eear
\end{definition}
\noindent The following theorem is adapted from \cite{At} to the super-setting, and it shows how the super Atiyah class is related to the existence of an affine connection on the complex supermanifold $\mani$.
\begin{theorem}[Pseudo-Atiyah] \label{PseudoAt} Let $\mani$ be a complex supermanifold and let $\cat{\emph{T}}_\mani$ be the tangent sheaf of $\mani$, then:
\begin{enumerate}[leftmargin=*] 
\item[(1)] the 1-jets short exact sequence \eqref{1jets} splits if and only if there exists an affine connection on $\mani$; 
\item[(2)] there exists an affine connection on $\mani$ if and only if $\mathfrak{At} (\cat{\emph{T}}_\mani)$ is trivial. 
%\bear
%\xymatrix@R=1.5pt{
%\mathfrak{At}_{\cat{\tiny{\emph{T}}}_\mani} : H^0 (|\mani|, \mathcal{H}om (\cat{\emph{T}}_\mani, \cat{\emph{T}}_\mani)) \ar[r] & H^1 (|\mani|, \mathcal{H}om (\cat{\emph{T}}_\mani, \cat{\emph{T}}_\mani ) \otimes \cat{\emph{T}}^\ast_\mani) \\
%id_{\cat{\tiny{\emph{T}}}^\ast_\mani \otimes \cat{\tiny{\emph{T}}}_\mani} \ar@{|->}[r] & \mathfrak{At} ({\cat{{\emph{T}}}_\mani}) \defeq \delta   (id_{\cat{\tiny{\emph{T}}}^\ast_\mani \otimes \cat{\tiny{\emph{T}}}_\mani} ) 
%}
%\eear
%whose image $\mathfrak{At}  (\cat{{\emph{T}}}_\mani )  \in H^1 (|\mani|, \mathcal{H}om (\cat{\emph{T}}_\mani, \cat{\emph{T}}_\mani ) \otimes \cat{\emph{T}}^\ast_\mani)$ is called \emph{Atiyah class} of $\cat{\emph{T}}_\mani$. In particular $\cat{\emph{T}}_\mani$ admits an affine connection if and only if $\mathfrak{At} (\cat{{\emph{T}}}_\mani )$ is trivial. 
\end{enumerate}
More in particular, letting $\{U_i\}_{i \in I}$ be an open covering for $|\mani_{\mathpzc{red}}|$ and $\{ g_{ij} \}_{i, j \in I}$ be the transition functions of $\cat{\emph{T}}_\mani$ on the intersections $U_i \cap U_j$, %with $g_{ij}$ the Jacobian of the change of coordinate, 
then the Atiyah class of $\cat{\emph{T}}_\mani$ is represented by the \v{C}ech $1$-cocycle 
\bear \label{AtiyahCocy}
\mathfrak{At} (\cat{\emph{T}}_\mani) \; \leftrightsquigarrow \;  \prod_{i < j}  \big (- (d_{\mathpzc{ev} }g_{ij} ) g_{ij}^{-1} \big )  \in H^1 (|\mani_{\mathpzc{red}}|, \cat{\emph{T}}^\ast_\mani \otimes_{\mathcal{O}_\mani} \mathcal{E}nd_{\mathcal{O}_{\mani}} (\cat{\emph{T}}_\mani)).
\eear
\end{theorem}
\begin{proof} The first point is the crucial one. First, we let $\nabla$ be an affine connection on $\mani$ and we define the morphism
%\bear
%\xymatrix@R=1.5pt{
%\sigma :  \cat{T}_\mani \ar[r] & \mathcal{J}^1 (\cat{T}_\mani) \\
%X \ar@{|->}[r] & (X, \nabla X) 
%}
%\eear
$s_\nabla : \cat{T}_\mani \rightarrow \mathcal{J}^1 \cat{T}_\mani$ by $j (X) \defeq (X , \nabla X).$ Notice that $s_{\nabla}$ is a well-defined morphism of sheaves of $\mathcal{O}_\mani$-modules, as
\bear
s_{\nabla} (f X) = ( f X , \nabla (fX) ) = (fX , d_{\mathpzc{ev}}f \otimes X + f \nabla X ) = f \cdot (X, \nabla X),
\eear 
by equation \eqref{OmodJet}. By definition of the 1-jets short exact sequence \eqref{1jets}, one has that the surjective morphism $\beta : \mathcal{J}^1 (\cat{T}_\mani) \rightarrow \cat{T}_\mani$ is given by $\beta ((X, \tau )) = X$. Therefore $\beta \circ s_{\nabla} = id_{\cat{\tiny{T}}_\mani}$, which implies that the affine connection $\nabla$ determines a splitting $s_{\nabla} $ of the 1-jets short exact sequence, \emph{i.e.}\
\bear
\xymatrix{
0 \ar[r] & \ar[r] \cat{T}^\ast_\mani \otimes_{\mathcal{O}_\mani} \cat{T}_\mani \ar[r]^\alpha & \mathcal{J}^1 \cat{T}_\mani \ar[r]_{\beta} & \cat{T}_\mani \ar@{->}@/_1.2pc/[l]^{s_\nabla} \ar[r] & 0.
}
\eear 
Vice versa, let the 1-jets short exact sequence \eqref{1jets} be split. Then there exists a morphism of sheaves of $\mathcal{O}_\mani$-modules $s: \cat{T}_\mani \rightarrow \mathcal{J}^1 \cat{T}_\mani$ such that $\beta \circ s = id_{\cat{\tiny{T}}_\mani}$. We let then $p: \mathcal{J}^1 \cat{T}_\mani \rightarrow \cat{T}^\ast_\mani \otimes_{\mathcal{O}_\mani} \cat{T}_\mani$ be defined by $p ((X, \tau)) = \tau$. Notice that $p $ is $\mathbb{C}$-linear, but not $\mathcal{O}_\mani$-linear. Let us then define $\nabla^{(s)} \defeq p \circ s : \cat{T}_\mani \rightarrow \cat{T}^\ast_\mani \otimes_{\mathcal{O}_\mani} \cat{T}_\mani.$ It is immediate that $\nabla^{(s)}$ is $\mathbb{C}$-linear. Finally, it satisfies Leibniz rule, indeed
\begin{align}
\nabla^{(s)} (fX) &= p ( s (fX)) = p (f s (X)) = p (f \cdot (X, \tau)) \nonumber =  p ((f X ,  d_{\mathpzc{ev}}f \otimes X + f \tau )) \\
& =  d_{\mathpzc{ev}}f \otimes X + f \tau = d_{\mathpzc{ev}}f \otimes X + f \nabla^{(s)}X, 
\end{align}
for any $f \in \mathcal{O}_X$ and $X \in \cat{T}_\mani.$ It follows that $\nabla^{(s)}$ defines an affine connection. \\
The second point of the Theorem depends on the first one. Let $\mathfrak{At} (\cat{{T}}_\mani) = 0$. Then, by definition $\delta (id_{ \cat{\tiny{{T}}}_\mani} ) = 0.$ By exactness, it follows from the 1-jets long cohomology exact sequence \eqref{1jetscoho} 
\bear \nonumber \hspace*{-.4cm}
\xymatrix@C=9pt{
\ldots \ar[r] &  H^0(|\mani_{\mathpzc{red}}|, \mathcal{H}om (\cat{T}_\mani,\mathcal{J}^1\cat{T}_\mani)) \ar[r] & H^0 (|\mani_{\mathpzc{red}}|, \mathcal{H}om (\cat{T}_\mani, \cat{T}_\mani)) \ar[r] & H^1 (|\mani_{\mathpzc{red}}|,  \cat{T}^\ast_\mani \otimes \mathcal{H}om (\cat{T}_\mani, \cat{T}_\mani) )  \ar[r] & \ldots
}
\eear
that there exists an element $h \in H^0 (|\mani_{\mathpzc{red}}|, \mathcal{H}om_{\mathcal{O}_\mani} (\cat{T}_\mani,\mathcal{J}^1(\cat{T}_\mani))$ such that $\beta \circ h = id_{ \cat{\tiny{{T}}}_\mani} $, where $\beta$ is the surjection in 1-jets short exact sequence \eqref{1jets}, which therefore splits. By the previous point of the Theorem, this is equivalent to the existence of an affine connection on $\mani.$\\
Viceversa, let $\mani$ be such that it admits an affine connection. Then the 1-jets short exact sequence is split by the previous point of the Theorem. This implies that there exists a map $h : H^0 (|\mani_{\mathpzc{red}}|, \mathcal{H}om_{\mathcal{O}_\mani} (\cat{T}_\mani,\mathcal{J}^1\cat{T}_\mani))$ such that $\beta \circ h = id_{\cat{\tiny{{T}}}_\mani}$. It follows that $id_{\cat{\tiny{{T}}}_\mani} $ belongs to the image of the map $H^0 (|\mani_{\mathpzc{red}}|, \mathcal{H}om_{\mathcal{O}_\mani} (\cat{T}_\mani,\mathcal{J}^1\cat{T}_\mani)) \rightarrow H^0 (|\mani_{\mathpzc{red}}|, \mathcal{H}om_{\mathcal{O}_\mani} (\cat{T}_\mani, \cat{T}_\mani))$ and hence $\delta (id_{\cat{\tiny{{T}}}_\mani}) = 0,$ \emph{i.e.}\ $\mathfrak{At} (\cat{{T}}_\mani) = 0$.\\
For the last point, we let $\{U_i \}_{i\in I}$ be an open covering of $\mani$ and $\{g_{ij} \}_{i, j \in I}$ the transition functions of $\cat{{T}}_\mani$ on the intersections $U_i \cap U_j$. For the sake of notation, compositions of maps are left understood in what follows. A \v{C}ech 1-cocycle representation of the class $\mathfrak{At} (\cat{T}_\mani)$ in terms of $\{g_{ij} \}_{i,j \in I} $ can be obtained by letting $\nabla_i $ be the (flat) connection on $\cat{T}_\mani |_{U_i}$ which is determined by a fixed trivialization relative to $\{ U_i\}_{i\in I}.$ In particular, following \cite{At} and \cite{BR} we let
\bear
\xymatrix@R=1.5pt{
\nabla_i : \cat{T}_\mani |_{U_i} \ar[r] & (\cat{T}^\ast_\mani \otimes_{\mathcal{O}_\mani} \cat{T}_\mani)|_{U_i}  \\
s \ar@{|->}[r] & \nabla_i s \defeq  \phi_i d_{\mathpzc{ev}} \phi^{-1}_i \, s,
}
\eear
where $\phi_i $ is a trivialization on $U_i$, and we define $ ( \mathfrak{a}_{ij} )_{i<j, i,j \in I} \in \Gamma (U_i \cap U_j , \cat{T}_\mani^\ast \otimes_{\mathcal{O}_{\mani} }\mathcal{E}nd_{\mathcal{O}_\mani} (\cat{T}_\mani) ) $ by 
\bear
\mathfrak{a}_{ij} = \nabla_j - \nabla_i.
\eear
Observing that $\phi_j = \phi_i \circ (\phi^{-1}_i \circ \phi_j) = \phi_i \circ g_{ji}^{-1} $, one computes
\begin{align}
\mathfrak{a}_{ij} & = \phi^{-1}_j d_{\mathpzc{ev}} \phi_j - \phi^{-1}_i d_{\mathpzc{ev}} \phi_i = \phi_i^{-1}  g_{ji}^{-1} d_{\mathpzc{ev}} g_{ji} \phi_i - \phi_i^{-1} d_{\mathpzc{ev}} \phi_i \nonumber \\
& = \phi^{-1}_{i} (g_{ji}^{-1} d_{\mathpzc{ev}} g_{ji} + g^{-1}_{ji}g_{ji} d_{\mathpzc{ev}}  - d_{\mathpzc{ev}}) \phi_i.
\end{align}
This simplifies to
\begin{align}
\mathfrak{a}_{ij} & = \phi^{-1}_i  (g^{-1}_{ji} (d_{\mathpzc{ev}}g_{ji}))  \phi_i  = \phi^{-1}_i  (g_{ij} (d_{\mathpzc{ev}}g^{-1}_{ij}))  \phi_i \nonumber \\
& = \phi^{-1}_i  (-(d_{\mathpzc{ev}}g_{ij}) g^{-1}_{ij})  \phi_i, 
\end{align}
where we have used Leibniz rule applied to $d_{\mathpzc{ev}} (g_{ij}g_{ij}^{-1}) = 0.$ It follows that 
\bear
\phi^{-1}_i \mathfrak{a}_{ij} \phi_i = -(d_{\mathpzc{ev}}g_{ij}) g^{-1}_{ij}.
\eear
Finally, upon using $g_{ij}g_{jk}g_{ki} = id_{\cat{\tiny{T}}_\mani}$, one checks that $\mathfrak{a}_{ij} \in Z^1 (\{ U_i\}_{i \in I}, \cat{T}_\mani^\ast \otimes_{\mathcal{O}_{\mani} }\mathcal{E}nd_{\mathcal{O}_\mani} (\cat{T}_\mani))$, \emph{i.e.}\ it defines a \v{C}ech 1-cocycle in the cohomology of $\cat{T}_\mani^\ast \otimes_{\mathcal{O}_{\mani} }\mathcal{E}nd_{\mathcal{O}_\mani} (\cat{T}_\mani)$.
\end{proof}

\begin{remark} 
With reference to the last part of Theorem \ref{PseudoAt}, one can notice that a local holomorphic connection can be written in the form $d_{\mathpzc{ev}} + \mathcal{A}_i $ in a trivialization $\phi_i : \pi^{-1} (U_i) \rightarrow U_i \times \mathbb{C}^{n|m},$ with $\mathcal{A}_i$ a matrix-valued holomorphic 1-form on $U_i$. These can be patched together to form a globally defined (holomorphic) affine connection if and only if
\begin{align} \label{condcech1}
\phi^{-1}_i (d_{\mathpzc{ev}} + \mathcal{A}_i ) \phi_i = \phi^{-1}_j (d_{\mathpzc{ev}} + \mathcal{A}_j) \phi_j,
%0 = (d_{\mathpzc{ev}} + A_j) - g_{ij}^{-1}(d_{\mathpzc{ev}} + A_i) g_{ij} = (d_{\mathpzc{ev}} + A_j) - (\phi_i \circ \phi_j^{-1})^{-1} (d_{\mathpzc{ev}} + A_i) (\phi_i \circ \phi_j), 
\end{align}
that can be rearranged as 
\bear \label{cechcomplete}
\phi^{-1}_j d_{\mathpzc{ev}} \phi_j - \phi^{-1}_i d_{\mathpzc{ev}} \phi_i = \phi^{-1}_i \mathcal{A}_i \phi_i - \phi^{-1}_j \mathcal{A}_j \phi_j.
\eear
Then, in view of Theorem \ref{PseudoAt}, the left-hand side is (a \v{C}ech 1-cocycle representing) the Atiyah class of $\cat{T}_\mani$, and equation \eqref{cechcomplete} can be written as
\bear
-(d_{\mathpzc{ev}}g_{ij}) g_{ij}^{-1} = \mathcal{A}_i - g_{ji}^{-1} \mathcal{A}_j g_{ji},
\eear
where the right-hand side is the \v{C}ech coboundary of $\mathcal{A}_i \in \Gamma (U_i, \cat{T}_\mani \otimes_{\mathcal{O}_\mani} \mathcal{E}nd(\cat{T}_\mani)).$ This shows via \v{C}ech cohomology that local connections can be patched together if and only if $\mathfrak{At}(\cat{T}_\mani) = 0,$ providing a different proof of the second point of Theorem \ref{PseudoAt} in a local-to-global fashion, as it is customary in \v{C}ech cohomology.
\end{remark}

\begin{remark}
Further, notice that the same construction as above can be carried out for any locally-free sheaf $\mathcal{E}$ on $\mani$, not only the tangent sheaf $\cat{T}_\mani$. On this respect the non-vanishing of the corresponding Atiyah class, which we still denote as $\mathfrak{At} (\mathcal{E})$, is an obstruction to define a holomorphic connection on $\mathcal{E}$.
\end{remark}

%Finally, observe that the Atiyah class of the tangent sheaf can be represented in terms of the transition functions of the tangent sheaf itself. If we let $\{U_i\}_{i \in I}$ be an open covering for $|\mani|$ and $\{ g_{ij} \}_{i, j \in I}$ be the transition functions of $\cat{T}_\mani$ on the intersections $U_i \cap U_j$ with $g_{ij}$ the Jacobian of the change of coordinate, then the Atiyah class of $\cat{T}_\mani$ is represented by the \v{C}ech $1$-cocycle 
%\bear \label{AtiyahCocy}
%\mathfrak{At} (\cat{T}_\mani) \; \leftrightsquigarrow \;  \prod_{i < j}  \big (- (dg_{ij} ) g_{ij}^{-1} \big )  \in H^1 (|\mani|, \cat{T}^\ast_\mani \otimes \mathcal{E}nd (\cat{T}_\mani)).
%\eear
\vspace{.3cm}
\noindent The previous Theorem \ref{PseudoAt} allows to identify the obstruction to splitting the $\Omega^1_{\tiny{\cat{M}}}$-extension. %More in details, we prove the following.
\begin{theorem}[$\Omega^{1}_{\tiny{\cat{M}}}$ and the Atiyah Class] \label{SplitAt} Let $\mani$ be a complex supermanifold and let $\cat{\emph{M}}$ be constructed as above. Then the $\Omega^1_{\tiny{\cat{\emph{M}}}}$-extension  
\bear \label{omegases}
\xymatrix{
0 \ar[r] & \pi^{\ast} \Omega^1_{\mani} \ar[r] & \Omega^{1}_{\tiny{{\cat{\emph{M}}}}} \ar[r] & \pi^\ast \emph{\cat{T}}_\mani  \ar[r] & 0
}
\eear
is split if and only if $\mathfrak{At} (\cat{\emph{T}}_\mani)$ is trivial.
In particular, the short exact sequence is split if and only if $\mani$ admits an affine connection. 
\end{theorem}
\begin{proof} By the previous Theorem \ref{projection} obstructions to splitting the short exact sequences lie indeed in $H^1 (|\mani_{\mathpzc{red}}|, \cat{T}^\ast_\mani \otimes_{\mathcal{O}_\mani} \mathcal{E}nd_{\mathcal{O}_{\mani}} (\cat{T}_\mani))$. By Lemma \ref{reductionlemma} and the structure of the transition functions given in Theorem \ref{transCotM} one sees that the obstructions are represented as \v{C}ech $1$-cocycles by elements of the form $ - (dg_{ij}) g^{-1}_{ij} $, (where the $g_{ij}$'s are the transition functions of the tangent sheaf $\cat{T}_\mani$), which is identified with the Atiyah class $\mathfrak{At} (\cat{T}_{\mani})$ by \eqref{AtiyahCocy}. %whose triviality is equivalent to the existence of an affine connection for $\mani$ by Theorem \ref{PseudoAt}.
\end{proof}
\noindent We now aim to relate the splitting of the short exact sequence \eqref{omegases} to the geometry of the complex supermanifold $\mani$. To this end we first recall some basic constructions specific to the theory of complex supermanifolds, see \cite{CNR} or \cite{Manin}. To each complex supermanifold is attached the short exact sequence
\bear \label{defses}
\xymatrix{
0 \ar[r] & \mathcal{J}_{\mani} \ar[r] & \mathcal{O}_{\mani} \ar[r] & \mathcal{O}_{\mani_{\mathpzc{red}}} \ar[r] & 0, 
}
\eear
where $\mathcal{J}_{\mani}$ is the sheaf of nilpotent sections in $\mathcal{O}_{\mani}$ and $\mathcal{O}_{\mani_{\mathpzc{red}}} = \mathcal{O}_\mani / \mathcal{J}_\mani$ is the structure sheaf of the reduced space $\mani_{\mathpzc{red}}$ - and ordinary complex manifold - of the supermanifold $\mani$. If \eqref{defses} splits, then the supermanifold $\mani$ is said to be \emph{projected}, because the splitting corresponds to the existence of a \virgolette projection'' morphisms $\pi : \mani \rightarrow \mani_{\mathpzc{red}} $ such that $\pi  \circ \iota = id_{\mani_{\mathpzc{red}}}$, if $\iota : \mani_{\mathpzc{red}} \hookrightarrow \mani$ is the canonical embedding of the reduced space $\mani_{\mathpzc{red}}$ into the supermanifold $\mani$. Moreover, the quotient $\mathcal{J}_{\mani } / \mathcal{J}^2_{\mani}$ defines a locally-free sheaf of $\mathcal{O}_{\mani_{\mathpzc{red}}}$-modules of rank $q$ - where $q$ is the odd dimension of $\mani$ - and whose section are seen to be odd. We call the quotient $\mathcal{J}_\mani / \mathcal{J}_\mani^2$ the \emph{fermionic sheaf} of $\mani$ and we denote with $\mathcal{F}_\mani$. 
We say that the supermanifold $\mani$ is \emph{split} if its structure sheaf is \emph{globally} isomorphic to the sheaf of exterior algebras $\wedge^\bullet \mathcal{F}_{\mani}$ over $\mathcal{O}_{\mani_{\mathpzc{red}}}$. %or, analogously, if $\mani$ is globally isomorphic to its so-called \emph{split model} $\mathfrak{S}(\manir, \mathcal{F}_\mani)$, constructed from the ordinary complex manifold $\manir$ and the sheaf of exterior algebras of the locally-free sheaf of $\mathcal{O}_{\manir}$-module $\mathcal{F}_\mani$. 
Notice that a split supermanifold is in particular projected. 
The corresponding obstruction theory to splitting a supermanifold is currently a compelling active research topic, see for example \cite{Betta} \cite{Noja}. \\
In this context, the \emph{fundamental obstruction class} to splitting a supermanifold $\mani$ is given by a class 
\bear
\omega_{\mani} \in H^1 (|\mani_{\mathpzc{red}}|, \mathcal{H}om_{\mathcal{O}_{\mani_{\mathpzc{red}}}} (\wedge^2 \mathcal{F}^\ast_\mani, \cat{T}_{\mani_{\mathpzc{red}}})) \cong H^1 (|\mani_{\mathpzc{red}}|, \cat{T}_{\mani_{\mathpzc{red}}} \otimes_{\mathcal{O}_{\mani_{\mathpzc{red}}}} \wedge^2 \mathcal{F}_\mani)
\eear
If $\omega_{\mani}$ is non-vanishing then $\mani$ is non-projected and in particular non-split. Whereas the fundamental obstruction class is always defined, higher obstruction classes 
\bear
H^1 (|\mani_{\mathpzc{red}}|, \mathcal{H}om_{\mathcal{O}_{\mani_{\mathpzc{red}}}} (\wedge^{2i+1} \mathcal{F}^\ast_\mani, \mathcal{F}^\ast_\mani)), \quad H^1 (|\mani_{\mathpzc{red}}|, \mathcal{H}om_{\mathcal{O}_{\mani_{\mathpzc{red}}}} (\wedge^{2i+2} \mathcal{F}^\ast_\mani, \cat{T}_{\mani_{\mathpzc{red}}}))
\eear
for $i\geq 1$ are defined if and only if all of the previous ones are vanishing, see the discussion in \cite{DonWit} for example. \\

%Whereas \emph{higher} obstruction classes to splitting a supermanifold might not be even defined, what counts to our purposes is that the so-called \emph{fundamental}  (or first) \emph{obstruction class} to splitting a supermanifold $\mani$ is always defined and it is given by a class $\omega_{\mani} \in H^1 (|\mani_{\mathpzc{red}}|, \mathcal{H}om (\wedge^2 \mathcal{F}^\ast_\mani, \cat{T}_{\mani_{\mathpzc{red}}})) \cong H^1 (|\mani_{\mathpzc{red}}|, \cat{T}_{\mani_{\mathpzc{red}}} \otimes \wedge^2 \mathcal{F}_\mani):$ if $\omega_{\mani}$ is non-vanishing then $\mani$ is non-projected and in particular non-split.\\
%{\bf ...storia sulle ostruzioni superiori $H^1 (|\mani_{\mathpzc{red}}|, \mathcal{H}om (\wedge^i \mathcal{F}^\ast_\mani, \mathcal{F}^\ast_\mani))$...}\\

\noindent A different criterion, actually a \emph{sufficient} condition, for the existence of a splitting of a supermanifold, has been given by Koszul in \cite{Kos}, relating the question about the splitting of $\mani$ to the existence of an affine connection on it.
\begin{theorem}[Koszul] \label{koszul} Let $\mani$ be a complex supermanifold. If $\mani$ admits an affine connection, then it is split. In particular, the affine connection defines a unique splitting of the supermanifold.
\end{theorem}
\begin{proof} see \cite{Kos}, recently reviewed in \cite{Betta}. 
\end{proof}
\noindent This result can in turn be related with a recent result by Donagi and Witten \cite{DonWit}. Indeed, when restricted to the reduced space, the tangent and cotangent sheaf $\cat{T}_\mani$ and $\cat{T}^\ast_\mani$ split into a direct sum of an even and an odd part. The latter is isomorphic to the fermionic sheaf or its dual in the case of the cotangent and tangent sheaf respectively. More precisely, one finds 
\begin{align}
\cat{T}_{\mani} |_{\mani_{\mathpzc{red}}} = \cat{T}_\mani \otimes_{\mathcal{O}_\mani} \mathcal{O}_{\mani_{\mathpzc{red}}} \cong \cat{T}_{\mani_{\mathpzc{red}}} \oplus \mathcal{F}^{\ast}_\mani,\\
\cat{T}^\ast_{\mani} |_{\mani_{\mathpzc{red}}} = \cat{T}^\ast_\mani \otimes_{\mathcal{O}_\mani} \mathcal{O}_{\mani_{\mathpzc{red}}} \cong \cat{T}^\ast_{\mani_{\mathpzc{red}}} \oplus \mathcal{F}_\mani.
\end{align} 
%In particular, posing $\cat{T}_{\mani} \lfloor_{\mani_{\mathpzc{red}}} \defeq \cat{T}_{\mani, + } \oplus \cat{T}_{\mani, -}$ one has $\cat{T}_{\mani, +} \cong \cat{T}_{\mani_{\mathpzc{red}}}$ and $\cat{T}_{\mani, -} \cong \mathcal{F}^{\ast}_\mani$, and similarly for $\cat{T}^\ast_\mani,$ the fundamental obstruction class take values in the cohomology group $H^1 (|\mani|, \mathcal{H}om (\wedge^2 \cat{T}_{\mani, -} , \cat{T}_{\mani, +})).$ 
In this spirit, one of the key result in \cite{DonWit} is concerned with the decomposition of the Atiyah class of $\cat{T}_\mani$ upon restriction of the tangent sheaf to the reduced manifold $\mani_{\mathpzc{red}}.$ 
\begin{theorem}[Donagi \& Witten] \label{donwit} Let $\mani$ be a complex supermanifold. Then, the restriction $\cat{\emph{T}}_{\mani} |_{\mani_{\mathpzc{red}}}$ of the tangent sheaf to $\mani_{\mathpzc{red}}$ induces the following decomposition of the cohomology group $H^1 (|\mani_{\mathpzc{red}}|, \mathcal{H}om_{\mathcal{O}_\mani} (S^2 \cat{\emph{T}}_\mani, \cat{\emph{T}}_\mani))$
\begin{align}
& H^1 (|\mani_{\mathpzc{red}}|, \mathcal{H}om_{\mathcal{O}_{\mani_{\mathpzc{red}}}} (\cat{\emph{S}}^2 \cat{\emph{T}}_\mani \lfloor_{\mani_{\mathpzc{red}}}, \cat{\emph{T}}_\mani \lfloor_{\mani_{\mathpzc{red}}})) \cong  \\
H^1 (|\mani_{\mathpzc{red}}|, \mathcal{H}om_{\mathcal{O}_{\mani_{\mathpzc{red}}}} (\cat{\emph{S}}^2 \cat{\emph{T}}_{\mani_{\mathpzc{red}}}, \cat{\emph{T}}_{\mani_{\mathpzc{red}}}))\,  \oplus \, & H^1 (|\mani_{\mathpzc{red}}|, \mathcal{H}om_{\mathcal{O}_{\mani_{\mathpzc{red}}}} (\wedge^2 \mathcal{F}_\mani, \cat{\emph{T}}_{\mani_{\mathpzc{red}}})) \oplus H^1 (|\mani_{\mathpzc{red}}|, \mathcal{H}om_{\mathcal{O}_{\mani_{\mathpzc{red}}}} ( \cat{\emph{T}}_{\mani_{\mathpzc{red}}} \otimes  \mathcal{F}_\mani, \mathcal{F}_\mani)). \nonumber
\end{align}
In particular, with respect to the above decomposition the Atiyah class $\mathfrak{At} (\cat{\emph{T}}_\mani)$ decomposes as follows
\bear \label{decompat}
\mathfrak{At} ( \cat{\emph{T}}_\mani ) \lfloor_{\mani_{\mathpzc{red}}} = \mathfrak{At} (\cat{\emph{T}}_{\mani_{\mathpzc{red}}}) \oplus \omega_{\mani} \oplus \mathfrak{At}( \mathcal{F}_\mani).
\eear 
where $\mathfrak{At} ( \cat{\emph{T}}_{\mani_{\mathpzc{red}}} )$ is the Atiyah class of the tangent sheaf of $\mani_{\mathpzc{red}}$, $\omega_{\mani}$ is first obstruction class, and $\mathfrak{At} (\mathcal{F}_\mani)$ is the Atiyah class of the fermionic sheaf.
\end{theorem}
\begin{proof} See \cite{DonWit}.
\end{proof}
\noindent This result together with Koszul's Theorem \ref{koszul} leads to the following for the geometry of $\Omega^1_{\tiny{\cat{M}}}$. 
\begin{theorem}[Splitting of $\Omega^{1}_{\tiny{\cat{M}}}$] \label{obsthm} Let $\mani$ be a complex supermanifold and let $\cat{\emph{M}}$ be constructed as above. Then any of the following is an obstruction to split $\Omega^{1}_{\tiny{\cat{\emph{M}}}}$-extension \eqref{omegases}:
\begin{enumerate}[leftmargin=*]
\item[(1)] $\mathfrak{At} (\cat{\emph{T}}_{\mani_{\mathpzc{red}}}) \neq 0$, \emph{i.e.}\ $\cat{\emph{T}}_{\mani_{\mathpzc{red}}}$ does not admit a holomorphic connection;
\item[(2)] $\mathfrak{At} (\mathcal{F}_\mani ) \neq 0$, \emph{i.e.}\ $\mathcal{F}_\mani$ does not admit a holomorphic connection;
\item[(3)] $\omega_{\mani} \neq 0$, \emph{i.e.} $\mani$ is non-projected or non-split.  
\end{enumerate}
In particular, a necessary condition for the $\Omega^{1}_{\tiny{\cat{\emph{M}}}}$-extension to split is that $\mani$ is a {split} supermanifold. 
\end{theorem}
\begin{proof} The $\Omega^1_{\tiny{\cat{M}}}$-extension is split if and only if the Atiyah class of $\cat{T}_\mani$ vanishes by Theorem \ref{SplitAt} and the three obstructions to split the $\Omega^1_{\tiny{\cat{M}}}$-extension follow from the decomposition at the Atiyah class given in \eqref{decompat} of Theorem \ref{donwit}. Finally, the vanishing of the Atiyah class implies the existence of an affine connection, which is equivalent to the existence of a splitting for $\mani$ by Theorem \ref{koszul}.
\end{proof}
%{\color{red} Not sure this is correct...!
%\begin{corollary}[Splitting of $\Omega^{1}_{\tiny{\cat{M}}}$ in dimension $n|2$] Let $\mani$ be a complex supermanifold of dimension $n|2$ and let $\cat{\emph{M}}$ be constructed as above from $\mani$. Then the $\Omega^{1}_{\tiny{\cat{\emph{M}}}}$-extension \eqref{omegases} splits if and only if 
%\begin{enumerate}[leftmargin=*]
%\item[(1)] $\mathfrak{At} (\cat{\emph{T}}_{\mani_{\mathpzc{red}}}) = 0$, \emph{i.e.}\ $\cat{\emph{T}}_{\mani_{\mathpzc{red}}}$ admits a holomorphic connection;
%\item[(2)] $\mathfrak{At} (\mathcal{F}_\mani ) = 0$, \emph{i.e.}\ $\mathcal{F}_\mani$ admits a holomorphic connection;
%\item[(3)] $\omega_{\mani} = 0$, \emph{i.e.} $\mani$ is split.  
%\end{enumerate}
%\end{corollary}}
\noindent We conclude this sections with some general remarks and speculation on the nature of the super Atiyah class for a complex supermanifold.   
\begin{remark}[Super Atiyah Class \& Super Characteristic Classes] It should be clear by the above considerations that the vanishing of the super Atiyah class provides a very strong constraint on the geometry of a complex or algebraic supermanifold. Namely, the following is an immediate consequence of the Koszul's result, Theorem \ref{koszul}, and the very definition of split supermanifold.
\begin{corollary} Let $\mani$ be a complex supermanifold of dimension $n|m$ such that $\mathfrak{At} (\cat{\emph{T}}_\mani) = 0,$ then $\mani $ is split. In particular all of the obstruction classes to splitting $\mani$ vanish, \emph{i.e.}\ for any $i = 1, \ldots, \lfloor m/2 \rfloor$
\bear
    H^1 (|\mani_{\mathpzc{red}}|, \mathcal{H}om_{\mathcal{O}_{\mani_{\mathpzc{red}}}} (\wedge^{2i} \mathcal{F}^\ast_\mani, \cat{\emph{T}}_{\mathpzc{red}} ) = 0, \quad  H^1 (|\mani_{\mathpzc{red}}|, \mathcal{H}om_{\mathcal{O}_{\mani_{\mathpzc{red}}}} (\wedge^{2i+1}\mathcal{F}_\mani, \mathcal{F}^\ast_\mani )) = 0.
\eear
%\begin{proof} Follows from Theorem \ref{koszul}, and the fact that a split supermanifold has vanishing obstruction classes by definition   
%\end{proof}
%\bear
%H^1 (|\mani_{\mathpzc{red}}|, \mathcal{H}om (\wedge^i \mathcal{F}^\ast_\mani, ( \cat{\emph{T}}_{\mani} \lfloor_{\mani_{\mathpzc{red}}})_{i \mbox{\tiny{\emph{mod}}} 2} ) = 0
%\eear 
\end{corollary}
\noindent This, together with Theorem \ref{donwit}, should make apparent the existence of a close relation between the super Atiyah class and the obstructions to splitting a complex supermanifold (see also the recent \cite{Betta} on this regard). \\
It should be stresed indeed that, in a classical setting, for a compact complex K\"{a}hler manifold $X$, the Atiyah class of tangent bundle contains informations about all the Chern classes $c_k (X) = c_{k} (\cat{T}_X)$ of the manifold \cite{Huy}. In particular, if $X$ admits a holomorphic connections, \emph{i.e.}\ the Atiyah class of $\cat{T}_X$ vanishes, then all the Chern classes vanishes as well, \emph{i.e.}\ $c_k (X) = 0$ for any $k>0$. It can be reasonably conjectured that, in a very similar fashion, also on a complex supermanifold $\mani$ the Atiyah class of the tangent bundle $\cat{T}_\mani$ contains informations about \emph{all} of the characteristic classes related to $\cat{T}_\mani$. The difference relies in that among these are to be counted not only the Atiyah classes of the reduced manifold $\mani_{\mathpzc{red}}$ and of the fermionic sheaf $\mathcal{F}_\mani$ - as seen in Theorem \ref{donwit} - and hence the related Chern classes, but also all of the obstruction classes to splitting the supermanifold, which indeed arises as well from the tangent bundle $\cat{T}_\mani$ of the supermanifold. Whereas the fundamental obstruction $\omega_\mani \in H^1 (|\mani_{\mathpzc{red}}|, \cat{T}_{\mani_{\mathpzc{red}}} \otimes_{\mathcal{O}_{\mani_{\mathpzc{red}}}} \wedge^2 \mathcal{F}_\mani)$ arises from the super Atiyah class upon restriction to $\mani_{\mathpzc{red}}$ \cite{DonWit}, it is quite reasonable to imagine that higher obstruction classes would arise upon allowing for higher fermionic terms in the filtration of $\mathcal{O}_\mani$ by the ideal sheaf of nilpotent sections $\mathcal{J}_\mani$ of the supermanifold
\bear
\mathcal{O}_{\mani_{\mathpzc{red}}} \subset \mathcal{O}_{\mani} / \mathcal{J}_\mani^2 \subset \ldots \subset \mathcal{O}_\mani / {\mathcal{J}_\mani^n} \subset \mathcal{O}_\mani,
\eear
and thus considering the restriction of $\cat{T}_\mani$ to the canonical superscheme contained in $\mani$ defined by the pair $\mani^{(\ell)} \defeq (|\mani_{\mathpzc{red}}|, \mathcal{O}_\mani / \mathcal{J}^\ell_\mani)$ for $\ell > 1$, \emph{i.e.} 
\bear
\cat{T}_\mani |_{\mani^{(\ell)}} = \cat{T}_\mani \otimes_{\mathcal{O}_\mani} \mathcal{O}_\mani / \mathcal{J}^\ell_{\mani} \cong \cat{T}_\mani /\mathcal{J}^\ell_\mani \cat{T}_\mani.
\eear

\end{remark}
%\noindent Remarkably, all of these instances occurs in actual examples. 

\section{Examples and Further Results} 

\noindent In this section we discuss and comment some examples. First, it is obvious that the complex supermanifold $\mathbb{C}^{n|m}$ admits a splitting for the $\Omega^1_{\tiny{\cat{M}}}$-extension for any values of $n$ and $m$. %Namely, we have the following immediate example.    
\begin{example}[$\mathbb{C}^{n|m}$] Let $\mathbb{C}^{n|m} \defeq (|\mathbb{C}^n|, \mathcal{O}_{\mathbb{C}^{n|m}})$ be the complex supermanifold with structure sheaf given by $\mathcal{O}_{\mathbb{C}^{n|m}} \defeq \mathcal{O}_{\mathbb{C}^n} \otimes \wedge^\bullet [ \theta] $ and let $\cat{M}$ be the supermanifold constructed from $\mathbb{C}^{n|m}$ as in Definition \ref{defM}. Then, for any $n$ and $m$ the $\Omega^1_{\tiny{\cat{M}}}$-extension is split.
\vspace{5pt}\\
This follows from the fact that the tangent bundle $\cat{T}_{\mathbb{C}^{n|m}}$ of $\mathbb{C}^{n|m}$ is trivial, and as such it admits a connection. Then, by Theorem \ref{PseudoAt} its super Atiyah class $\mathfrak{At} (\cat{T}_{\mathbb{C}^{n|m}})$ is trivial, and in turn, by Theorem \ref{SplitAt} the $\Omega^1_{\tiny{\cat{M}}}$-extension related to $\mathbb{C}^{n|m}$ splits. %The same results holds true also in the algebraic category, using Zariski topology.
\end{example} 
%\noindent The result follows immediately from the fact that any higher-cohomology group vanishes on $\mathbb{C}^n$. Notice that the same result applies to the algebraic category substituting $\mathbb{C}^{n|m}$ with the algebraic supermanifold $\mathbb{A}_{\mathbb{C}}^{n|m} = (|\mathbb{A}^n_\mathbb{C}|, \mathcal{O}_{\mathbb{A}_{\mathbb{C}}^{n|m}})$ for $\mathbb{A}^n_{\mathbb{C}} = \mathfrak{spec} ( \mathbb{C}[x_1, \ldots, x_n])$ and $\mathcal{O}_{\mathbb{A}^{n|m}_\mathbb{C}} \defeq \mathbb{C}[x_1, \ldots, x_n] \otimes \wedge^\bullet [ \theta]$, since any sheaf on the affine space $\mathbb{A}^n_{\mathbb{C}}$ has vanishing higher-cohomology. \\
\noindent A way more interesting example is provided by complex Lie supergroups - the reader is suggested to refer to \cite{Vish} for the relevant definitions. 

\begin{example}[Complex Lie Supergroups $\mathpzc{G}$] Let $\mathpzc{G}$ be a complex Lie supergroup. Then $\mathpzc{G}$ admits a holomorphic connection and hence it is split. In particular, if $\cat{M}$ is the supermanifold constructed from $\mathpzc{G}$ as in Definition \ref{defM}, the $\Omega^1_{\tiny{\cat{M}}}$-extension splits.
\vspace{5pt}\\
The result follows from the fact that, as in the ordinary theory, a complex Lie supergroup is parallelizable, \emph{i.e.}\ its tangent bundle $\cat{T}_{\mathpzc{G}}$ is trivial. Just like in the ordinary theory this depends on the existence of a group structure on the supermanifold $\mathpzc{G}$. In turn, since the tangent bundle $\cat{T}_{\mathpzc{G}}$ of $\mathpzc{G}$ is trivial, then $\mathpzc{G}$ admits a connection. It follows from Theorem \ref{koszul} that $\mathpzc{G}$ is split and from Theorem \ref{PseudoAt} and Theorem \ref{SplitAt} that the related $\Omega^1_{\tiny{\cat{M}}}$-extension splits.
\end{example}
\noindent It is to be stressed that the case of homogeneous supermanifolds, \emph{i.e.}\ quotients of Lie supergroups by some closed Lie sub-(super)group, is more delicate: indeed, complex homogeneous supermanifolds can indeed be non-split, thus not admitting holomorphic connection, see \cite{Vish}.\\
We now move to (complex) projective superspaces $\mathbb{CP}^{n|m}$, which are defined as the complex supermanifolds given by the pair 
$
\mathbb{CP}^{n|m} \defeq (|\mathbb{CP}^n|, \mathcal{O}_{\mathbb{CP}^{n|m}} )
$
where $\mathcal{O}_{\mathbb{CP}^{n|m}} \defeq \wedge^\bullet_{\mathcal{O}_{\mathbb{CP}^n}}(\Pi \mathcal{O}_{\mathbb{CP}^{n}} (-1)^{\oplus m})$. Notice that according to the discussion after Theorem \ref{SplitAt} above, complex projective superspaces $\mathbb{CP}^{n|m}$ are by definition split supermanifolds, in particular their fermionic sheaf reads $\mathcal{F}_{\mathbb{CP}^{n|m}} \defeq \Pi \mathcal{O}_{\mathbb{C}^{n|m}} (-1)^{\oplus m}.$ 
See \cite{CN} for a dedicated paper. %It is immediate to see that no projective superspace $\mathbb{CP}^{n|m}$ admits a splitting for the $\Omega^1_{\tiny{\cat{M}}}$-extension, as a consequence of the well-known fact that no projective space $\mathbb{CP}^n$ for $n\geq 1$ admits a holomorphic affine connection for $n\geq 1$. We state this fact as a lemma.
%first proving that the $\Omega^{1}_{\tiny{\cat{{M}}}}$-extension does not split whenever $\cat{M}$ is constructed from any (complex) \emph{projective superspace} $\mathbb{CP}^{n|m}$. The interested reader can refer to \cite{CN} for the supergeometry of projective superspace $\mathbb{CP}^{n|m}$.
\begin{example}[$\mathbb{CP}^{n|m}$] \label{CPAt} Let $\mathbb{CP}^{n|m}$ any complex projective superspace and let $\cat{{M}}$ be the supermanifold constructed from $\mathbb{CP}^{n|m}$ as above. Then, for any value of $n\geq 1$ and $m\geq0$ the $\Omega^1_{\tiny{\cat{{M}}}}$-extension is not split.
\vspace{5pt}\\
The result can be seen to follows from point (1) of Theorem \ref{obsthm}. For this, one needs to prove that projective spaces $\mathbb{CP}^{n}$ do not admit affine holomorphic connections. 
For $n \geq 1$ the (dual of the) Euler exact sequence tensored by $\mathcal{E}nd_{\mathcal{O}_{\mathbb{CP}^n}}(\cat{T}_{\mathbb{CP}^n})$ reads
\bear \label{eulertwisted}
\xymatrix{
0 \ar[r] & (\cat{T}^\ast_{\mathbb{CP}^n})^{\otimes 2} \otimes \cat{T}_{\mathbb{CP}^n } \ar[r] & (\cat{T}_{\mathbb{CP}^n}^{\ast}  \otimes \cat{T}_{\mathbb{CP}^n} (-1))\otimes{\mathbb{C}}^{\oplus n+1} \ar[r] & \cat{T}^\ast_{\mathbb{CP}^n} \otimes \cat{T}_{\mathbb{CP}^n} \ar[r] & 0.
}
\eear
Here all tensor products are over $\mathcal{O}_{\mathbb{CP}^{n}}.$ In the case $n > 1$ one computes
%This is nothing but the (dual) Euler exact sequence twisted by $\cat{T}^\ast_{\mathbb{CP}^n} \otimes \cat{T}_{\mathbb{CP}^n}$. Appropriate use of (twisted) Euler exact sequence yields that 
\bear
H^0 (\mathbb{CP}^n , (\cat{T}_{\mathbb{CP}^n}^{\ast}  \otimes \cat{T}_{\mathbb{CP}^n} (-1))) =  0, \qquad H^1 (\mathbb{CP}^n, \cat{T}_{\mathbb{CP}^n}^{\ast}  \otimes \cat{T}_{\mathbb{CP}^n} (-1)) = \mathbb{C}^{n+1}, 
\eear
\bear
H^0 (\mathbb{CP}^n , (\cat{T}_{\mathbb{CP}^n}^{\ast}  \otimes \cat{T}_{\mathbb{CP}^n} ) =  \mathbb{C}, \qquad H^1 (\mathbb{CP}^n, \cat{T}_{\mathbb{CP}^n}^{\ast}  \otimes \cat{T}_{\mathbb{CP}^n} (-1)) = 0.
\eear
It follows that long cohomology sequence associated to \eqref{eulertwisted} reads
\bear
\xymatrix{
0 \ar[r] & H^0 (\mathbb{CP}^n , \cat{T}_{\mathbb{CP}^n}^{\ast}  \otimes \cat{T}_{\mathbb{CP}^n} ) \cong \mathbb{C} \ar[r]^{ \delta \qquad } & H^1 (\mathbb{CP}, (\cat{T}^\ast)_{\mathbb{CP}^n}^{\otimes 2} \otimes \cat{T}_{\mathbb{CP}^n }) \cong \mathbb{C}^{(n+1)^2 + 1} \ar[r] & \ldots,
}
\eear
and injectivity of the map implies that $\mathfrak{At} (\cat{T}_{\mathbb{CP}^n}) \neq 0.$ In the case of $\mathbb{CP}^1$, one has that $\cat{T}^\ast_{\mathbb{CP}^1} \cong \mathcal{O}_{\mathbb{CP}^1} (-2)$ and $\cat{T}_{\mathbb{CP}^1} \cong \mathcal{O}_{\mathbb{CP}^1} (+2).$ The long cohomology exact sequence reads
\bear
\xymatrix{
0 \ar[r] & H^0 (\mathbb{CP}^1 , \cat{T}_{\mathbb{CP}^1}^{\ast}  \otimes \cat{T}_{\mathbb{CP}^1} ) \cong \mathbb{C} \ar[r]^{\; \delta} & H^1 (\mathbb{CP}^1, (\cat{T}^\ast_{\mathbb{CP}^1})^{\otimes 2} \otimes \cat{T}_{\mathbb{CP}^1 }) \ar[r] & 0.
}
\eear
Alternatively, one can just observe that for a line bundle the Atiyah class equals the Chern class, \emph{i.e.}\ $\mathfrak{At} (\cat{T}_{\mathbb{CP}^1}) = c (\cat{T}_{\mathbb{CP}^1})$ and $c (\cat{T}_{\mathbb{CP}^1}) = c_1 (\cat{T}_{\mathbb{CP}^1}) = 2,$ see \cite{Huy}. %and the above isomorphism of $\mathbb{C}$-vector spaces reads explicitly $
%\mathbb{C} \cdot \big \langle dz \otimes \partial_{z} \big \rangle \mapsto \mathbb{C} \cdot \big \langle \frac{{dz}^{\otimes 2} \otimes \partial_{z}}{z} \big \rangle. $ 
It then follows from Theorem \ref{obsthm} that the $\Omega^1_{\tiny{\cat{M}}}$-extension related to $\mathbb{CP}^{n|m}$ does not split. 
\end{example}
\begin{remark} Note that in the ordinary complex geometric setting, the vanishing of the Atiyah class $\mathfrak{At}({E}) $ of a certain (Hermitian holomorphic) vector bundle ${E}$ on a compact complex manifold $X$ is equivalent for the vector bundle $E$ to be flat. More precisely, the class of the curvature $[\mathcal{F}_{\nabla^\mathpzc{C}_{E}}] \in H^1 (X, \cat{T}^\ast_X \otimes_{\mathcal{O}_X} \mathcal{E}nd (E))$ of the Chern connection $\nabla_{E}^\mathpzc{C}$ of the vector bundle $E$ corresponds to its Atiyah class $\mathfrak{At} (E)$. This result is by no mean true in a complex supergeometric setting. Indeed all of the Calabi-Yau projective superspaces $\mathbb{CP}^{n|n+1}$ for $n\geq1$ admit a flat Chern connection via a generalization of the Fubini-Study metric \cite{NojaCY}, but none of them have vanishing super Atiyah class as shown above. %This is a typical example of how things are subtler than expected in complex supergeometry. 
\end{remark}
\noindent Another interesting example is provided by the possibly easiest - yet non-trivial - complete intersection in $\mathbb{CP}^{2|2}$, corresponding to the super conic $\mathcal{C} \subset \mathbb{CP}^2$ cut out by the equation 
\bear \label{conicS}
X_0^2 + X_1^2 + X_2^2 + \Theta_1 \Theta_2 = 0  \quad \subset \quad \mathbb{CP}^{2|2}.
\eear   
The supergeometry related to the above equation in projective superspace is non-trivial and not that well-known outside an audience of experts in supergeometry. To help the reader make sense and appreciate this example, we have opted to discuss some of the more peculiar aspects of the geometry of the super conic in Appendix \ref{app2}. We thus refer to it for more informations. Here we content ourself to say that the resulting $1|2$-dimensional supermanifold is notably isomorphic to the non-projected supermanifold constructed out of the following three data $(\mathcal{C}_{\mathpzc{red}} = \mathbb{CP}^1, \mathcal{F}_\mathcal{C} = \mathcal{O}_{\mathbb{CP}^1} (-2)^{\oplus 2}, \omega_\mathcal{C} = 1)$, where the non-zero obstruction class $\omega_\mathcal{C} \in H^1 (\mathbb{CP}^1, \cat{T}_{\mathbb{CP}^1} \otimes \wedge^2 \mathcal{F}_\mani)$ is seen via the isomorphism $ H^1 (\mathbb{CP}^1, \cat{T}_{\mathbb{CP}^1} \otimes \wedge^2 \mathcal{F}_\mathcal{C}) \cong H^1 (\mathbb{CP}^1, \mathcal{O}_{\mathbb{CP}^1} (-2)) \cong \mathbb{C}$. %It is to be noted that the (even) Picard group $\mbox{Pic}_0 (\mathcal{C})$ has a continuous part - it is isomorphic to $\mathbb{Z} \oplus \mathbb{C}$ - due to the non-projected nature of $\mathcal{C}$ and indeed the above embedding in $\mathbb{CP}^{2|2}$ is realized via the global sections of a line bundle on $\mathcal{C} $ which has a non-trivial continuos component in $\mbox{Pic}_0 (\mathcal{C}).$ 
We have already seen that the Atiyah class of $\mathbb{CP}^1$ is non-zero in Example \ref{CPAt}, moreover the rank 2 vector bundle $\mathcal{F}_\mathcal{C} = \mathcal{O}_{\mathbb{CP}^1} (-2) \oplus \mathcal{O}_{\mathbb{CP}^1} (-2)$ is obviously not flat, and hence it has a non-trivial Atiyah class. We can thus conclude that $\Omega^1_{\tiny{\cat{M}}}$-extension related to the super conic $\mathcal{C}$ is totally obstructed in the sense of Theorem \ref{obsthm}, \emph{i.e.}\ all of the obstructions in the points (1)-(3) in the statement of Theorem \ref{obsthm} are non-zero. We summarize this in the following. 
 \begin{example}[Super Conic $\mathcal{C}$] \label{SuperConic} Let $\mathcal{C}$ be the complete intersection given by the equation $X_0^2 + X_1^2 + X_2^2 + \Theta_1 \Theta_2 = 0$ in $\mathbb{CP}^{2|2}$ and let $\cat{{M}}$ be the supermanifold constructed from $\mathcal{C}$ as above. Then the related $\Omega^1_{\tiny{\cat{{M}}}}$-extension is not split.
\end{example}
\noindent Finally, let us consider a slightly trickier example, namely that of a super elliptic curve $\mathpzc{E}$ of dimension $1|3$ modeled on an ordinary elliptic curve $\mathpzc{E}_{\mathpzc{red}} = E$ and whose rank 3 fermionic sheaf is given by the direct sum $\mathcal{F}_{\mathpzc{E}} = \mathcal{O}_E^{\oplus 3}$. %Note that in this case Theorem \ref{dim2} does not in general provide 
We further assume that the fundamental obstruction class $\omega_\mathpzc{E} \in H^1 (E, \cat{T}_E \otimes \wedge^2 \mathcal{O}_E^{\oplus 3}) $ is vanishing, \emph{i.e.}\ $\omega_\mathpzc{E} = (0, 0, 0)$ in the isomorphism $H^1 (E, \cat{T}_E \otimes \wedge^2 \mathcal{O}_E^{\oplus 3}) = H^1 (E, \mathcal{O}_E^{\oplus 3}) \cong \mathbb{C}^3$. Under these hypotheses one has that all of the points (1)-(3) in Theorem \ref{obsthm} are indeed satisfied since also $\mathfrak{At} (\mathpzc{E}_\mathpzc{red}) = 0 = \mathfrak{At} (\mathcal{F}_\mathpzc{E})$. % indeed the first Chern class (which equals the Atiyah class) of the tangent bundle of $\mathpzc{E}_{\mathpzc{red}} = E$ is trivial and the bundle $\mathcal{O}_{E}^{\oplus 3}$ admits a holomorphic connection. 
 Nonetheless, the $\Omega^1_{\tiny{\cat{M}}}$-extension related to $\mathpzc{E}$ might still be non-split, since $\mathpzc{E}$ is projected but not necessarily split as a complex supermanifold. Indeed the higher obstruction to split $\mathpzc{E}$ takes values in the cohomology group $H^1 (E, \mathcal{F}^\ast_\mathpzc{E} \otimes \wedge^3 \mathcal{O}_E^{\oplus 3}),$ which is computed to be isomorphic to $\mathbb{C}^{\oplus 3}$. A non-vanishing class obstruction class $\omega^{(3)}_{\mathpzc{E}} = (c_1, c_2, c_3) \neq 0$ would correspond to odd transition functions of the kind $\theta_i \mapsto \theta_i + c_i \theta_1 \theta_2 \theta_3$ for some complex number $c_i$. In this case, if $\mathpzc{E}$ is a non-split supermanifold, the $\Omega^1_{\tiny{\cat{M}}}$-extension related to $\mathpzc{E}$ is also non-split. The upshot of the example is that in the case of supermanifolds of odd dimension greater than $2$, the presence of higher obstruction classes to split a supermanifold $\mani$ is quite a delicate issues.  This is summarized in the following. 
 \begin{example}[Super Elliptic Curve of dimension $1|3$] Let $\mathpzc{E}$ be a supermanifold of dimension $1|3$ such that $\mathpzc{E}_{\mathpzc{red}} = E$ for $E$ an elliptic curve (over $\mathbb{C}$) and $\mathcal{F}_{\mathpzc{E}} = \mathcal{O}^{\oplus 3}_{{E}}$ and let its fundamental obstruction class $\omega_\mathpzc{E} \in H^1 (E, \cat{T}_E \otimes  \wedge^2 \mathcal{O}^{\oplus 3}_{E} ) $ be zero. Then the higher obstruction to split $\mathpzc{E}$
\bear
\omega^{(3)}_{\mathpzc{E}} \in H^1 (E, \mathcal{F}^\ast_\mathpzc{E} \otimes \wedge^3 \mathcal{O}_E^{\oplus 3}) \cong H^1 (E, \mathcal{O}^{\oplus 3}_E) \cong \mathbb{C}^{\oplus 3},
\eear
is defined and fully determines the geometry of $\mathpzc{E}$. In particular, $\omega^{(3)}_\mathpzc{E}$ is also an obstruction to split the the $\Omega^1_{\tiny{\cat{M}}}$-extension related to $\mathpzc{E}$, \emph{i.e.} if $\mathpzc{E}$ is non-split, then the $\Omega^1_{\tiny{\cat{M}}}$-extension does not split as well. 
\end{example}

%{\bf Examples on curves of any genus...? At least elliptic curve? It would be nice to show some higher-obstruction due to non-splitness in the case of a $1|2$, but also $1|3$ elliptic curve... In particular: in the case $1|3$ it would be nice to see how the full Atiyah class look...}

\section{Local Theory: Forms, Natural Operators and Cohomology}

\noindent In this section we study forms and natural, \emph{i.e.}\ globally defined and invariant, operators acting on $\Omega^\bullet_{\tiny{\cat{M}}}$ in the smooth and holomorphic category. We start by making contact between our framework and the setting developed by the author and collaborators in \cite{CNR}, where {differential} and {integral} forms on a real or complex supermanifold $\mani$ are recovered in a unified fashion starting from the {triple} tensor product of natural sheaves on $\mani$ given by $\Omega^\bullet_{\mani} \otimes_{\stsheaf} \mathcal{D}_\mani \otimes_{\stsheaf} (\Omega^\bullet_{\mani})^\ast,$ where $\mathcal{D}_\mani$ is the \emph{sheaf of differential operators} on $\mani$. Notice that due to the pivotal presence of $\mathcal{D}_\mani$ this is a {non-commutative} construction, better than just super-commutative. \\
This triple tensor product is acted upon by two globally defined {mutually commuting} operators $\hat d$ and $\hat \delta$, see \cite{CNR}
\bear
\begin{tikzcd}
\arrow[loop left, "\hat d"]\qquad \qquad \qquad \Omega^\bullet_{\mani} \otimes_{\stsheaf} \mathcal{D}_\mani \otimes_{\stsheaf} (\Omega^\bullet_{\mani})^\ast \qquad \qquad \qquad \arrow[loop right, "\hat \delta"] 
\end{tikzcd}
\eear
and as such it yields a double complex of sheaves, $(\Omega^\bullet_{\mani} \otimes_{\stsheaf} \mathcal{D}_\mani \otimes_{\stsheaf} (\Omega^\bullet_{\mani})^\ast , \hat \delta, \hat d)$. Differential and integral forms, together with their differentials and {Poincaré lemmas} are recovered via the two natural spectral sequences - we call them $E^{\hat d}_i$ and $E^{\hat \delta}_i$, depending on whether we are starting computing the cohomology with respect of $\hat d $ or $\hat \delta$ - related to this double complex $(\Omega^\bullet_{\mani} \otimes_{\stsheaf} \mathcal{D}_\mani \otimes_{\stsheaf} (\Omega^\bullet_{\mani})^\ast , \hat \delta, \hat d)$. In particular, the following holds true.
\begin{theorem}[\cite{CNR}] Let $\mani$ be a real or complex supermanifold. Then the spectral sequences $E_i^{\hat \delta}$ and $E_i^{\hat d}$ related to the double complex $(\Omega^\bullet_{\mani} \otimes_{\stsheaf} \mathcal{D}_\mani \otimes_{\stsheaf} (\Omega^\bullet_{\mani})^\ast, \hat \delta, \hat d)$ 
\begin{enumerate}[leftmargin=*]
\item yield the \emph{differential forms} and \emph{integral forms} on $\mani$ at page 1, \emph{i.e.}\
\bear
E^{\hat \delta}_{1} \cong \Omega^\bullet_{\mani}, \qquad E_1^{\hat d} \cong \mathcal{B}er (\mani) \otimes_{\stsheaf} (\Omega^\bullet_{\mani})^\ast;
\eear
\item both converge to the locally constant sheaf $\mathbb{K}_{\mani}$ for $\mathbb{K}$ the real or complex numbers at page 2, \emph{i.e.}\ 
\bear
E^{\hat \delta}_2  = E^{\hat \delta}_\infty = \mathbb{K}_{\mani}, \qquad E^{\hat d}_2 = E^{\hat d}_\infty \cong \mathbb{K}_\mani.
\eear
\end{enumerate}
\end{theorem}
\noindent A consequence of the above is that the hypercohomologies of differential and integral forms are isomorphic, and both coincides with the Rham cohomology of the reduced manifold, see also the recent \cite{Noja}.
\begin{corollary}[\cite{CNR}] \label{quasi} Let $\mani$ be a real supermanifold. Then the hypercohomologies of (the sheaf of) differential forms $H^\bullet_{\mathpzc{dif}} (\mani)$ and integral forms $ H^\bullet_{\mathpzc{int}} (\mani)$ are isomorphic. In particular, one finds 
\bear
H^\bullet_{\mathpzc{dif}} (\mani) \cong \check{H}^\bullet (|{\mani_{\mathpzc{red}}}|, \mathbb{R}_{\mani}) \cong H^\bullet_{\mathpzc{int}} (\mani).
\eear
  \end{corollary}
\noindent One the main ingredient of the above construction is the non-commutative sheaf of differential operator $\mathcal{D}_\mani$. To $\mathcal{D}_\mani$ is canonically associated a sheaf of {super-commutative} $\mathcal{O}_\mani$-algebras, by considering the filtration $\mathcal{D}_\mani^{(\leq i)} \subseteq \mathcal{D}_\mani^{(\leq i+1)}$ by the {degree} of the differential operators for any $i \geq 0$. This is given by the quotient
\bear
\mbox{gr}^\bullet (\mathcal{D}_\mani) \defeq \bigoplus_{i = 0}^\infty \slantone{\mathcal{D}^{(\leq i)}_{\mani }}{\mathcal{D}^{(\leq i-1)}_{\mani}}.
\eear 
It is not hard to see that $\mbox{gr}^\bullet (\mathcal{D}_\mani ) \cong \cat{S}^\bullet \cat{T}_\mani$. In this way, a \emph{\virgolette de-quantization''} of the above triple tensor product reads
\bear \label{dequant} \xymatrix{
\Omega^\bullet_{\mani} \otimes_{\stsheaf} \mathcal{D}_\mani \otimes_{\stsheaf} (\Omega^\bullet_{\mani})^\ast \ar@{~>}[rr]^{\; \mbox{\emph{\tiny{de-quantization}}} \quad} && \Omega^\bullet_{\mani} \otimes_{\stsheaf} \cat{S}^\bullet \cat{T}_\mani  \otimes_{\stsheaf} (\Omega^\bullet_{\mani})^\ast.
}
\eear
This sheaf of super-commutative $\mathcal{O}_\mani$-algebras can be put in relation with the sheaf $\Omega^\bullet_{\cat{\tiny{M}}}$, seen as a sheaf of $\mathcal{O}_\mani$-modules. Indeed one can observe that the $\Omega^1_{\cat{\tiny{M}}}$-extension \eqref{extomega} is always locally split, so that over an open set $\pi^{-1} (U)$ of $\cat{M}$ for $U$ an open set in $\mani$, one has 
\bear
\Omega^1_{\pi^{-1} (U)} \cong \pi^\ast \Omega^1_{U} \oplus \pi^\ast \cat{T}_U.
\eear
This holds true globally for a real supermanifold, as proved above. It follows that %one finds the decomposition  
\begin{align} \label{deco}
\Omega^\bullet_{\pi^{-1}(U)} & % \cat{S}^\bullet_{\mathcal{O}_{\pi^{-1} (U)}} \left ( \pi^\ast \Omega^1_{U} \oplus \pi^\ast \cat{T}_U \right ) 
\cong \cat{S}^\bullet_{\mathcal{O}_{\mani}} \left ( \pi^\ast \Omega^1_{U} \oplus \pi^\ast \cat{T}_U \right ) \otimes \mathcal{O}_{\pi^{-1} (U)}
%& \cong \cat{S}^\bullet \Omega^1_{U} \otimes \cat{S}^\bullet \cat{T}_{U} \otimes (\Omega^\bullet_{U})^\ast \\
 \cong \Omega^\bullet_{U} \otimes_{\mathcal{O}_U} \cat{S}^\bullet \cat{T}_U \otimes_{\mathcal{O}_U} (\Omega^\bullet_{U})^\ast,
\end{align}
to be compared to \eqref{dequant} above. Using the decomposition \eqref{deco}, the action of the \emph{de Rham differential} $ \cat{d} : \Omega^{\bullet}_{\tiny{\cat{M}}} \rightarrow \Omega^{\bullet}_{\tiny{\cat{M}}}$ can be given in an open set $\pi^{-1} (U)$ with local coordinate $ x_a$ and $p_a$ as follows:
\begin{align} \label{deRhamdiff}
%\cat{d} (\omega \otimes F \otimes \tau ) &
 \cat{d} (\eta \otimes F \otimes f(x, p)) & =  (-1)^{|\eta| + |F|} \left ( \eta \otimes F \otimes dx_a \frac{\partial f}{\partial x_a} + \eta \otimes F \otimes dp_a \frac{\partial f}{\partial p_a} \right )  \\
& = (-1)^{|F||x_a| + |\eta|} \eta dx_a \otimes F \otimes \frac{\partial f}{\partial x_a} + (-1)^{|\eta| + |F|} \eta \otimes F dp_a \otimes \frac{\partial f}{\partial p_a}, \nonumber
\end{align}
where $\eta \in \Omega^\bullet_\mani$, $F \in \cat{S}^\bullet \cat{T}_\mani$ and $f \in (\Omega^\bullet_\mani )^\ast = \mathcal{O}_{\cat{\tiny{M}}}.$ The sum over $a$ is left understood. %The previous local expression \eqref{deRhamdiff} for $\cat{d}$ yields immediately the following result. 
\begin{theorem}[Homology of \cat{d} / Poincaré Lemma] Let $\cat{\emph{M}}$ be defined as above and let $\cat{\emph{d}} : \Omega^{\bullet}_{\tiny{\cat{\emph{M}}}} \rightarrow \Omega^{\bullet}_{\tiny{\cat{\emph{M}}}}$ the de Rham differential. Then 
\bear
H_{\tiny{\cat{\emph{d}}}} (\Omega^\bullet_{\tiny{\cat{\emph{M}}}}) \cong \mathbb{K}_{\tiny{\cat{\emph{M}}}},
\eear
where $\mathbb{K}_{\tiny{\cat{\emph{M}}}}$ is the sheaf of locally-constant functions on $\cat{\emph{M}}$ for $\mathbb{K}$ the real or complex numbers.
\end{theorem}
\begin{proof} Given the action of the de Rham differential $\cat{d}$ in \eqref{deRhamdiff}, the result follows from the ordinary Poincaré lemma for supermanifolds, see for example \cite{Noja}. \end{proof}

\subsection{Odd Symplectic Form and its Cohomology} Let us keep working in the smooth or holomorphic category and let us now consider the (non-degenerate) odd 2-form $\omega = \sum_a dx_a dp_a \in (\Omega^2_{\tiny{\cat{M}}})_1$ where the index $a$ runs over both even and odd coordinates. We first observe the following.
\begin{lemma}[Global Definition on $\omega$] The odd 2-form $\omega = \sum_a dx_a dp_a \in (\Omega^2_{\tiny{\cat{\emph{M}}}})_1$ is invariant, \emph{i.e.}\ coordinate independent. 
\end{lemma}
\begin{proof} We use the coordinate transformations of Lemma \ref{transCotM}. %passing from a system of local coordinates $(x_a, p_a)$ to another system of local coordinates given by $(z_a, q_a)$. 
One finds that 
\begin{align}
dx_a \left ( \frac{\partial z_b}{\partial x_a} dq_b  + (-1)^{|x_a| + |z_b|} d\left ( \frac{\partial z_b}{\partial x_a} \right ) q_b \right )  = dz_a dq_a + (-1)^{|x_a| + |z_b|} dx_a dx_c \left ( \frac{\partial^2 z_b}{\partial x_c \partial x_a} \right ) q_b,
\end{align}
and it is easy to verify that the contribution of the second term is zero.  
\end{proof}
\noindent Another way to see that $\omega$ is actually invariant is to introduce its \emph{primitive form}
$
\eta \defeq (-1)^{|x_a| + 1} dx_a p_a.
$
\begin{lemma}[Primitive Form of $\omega$] The primitive form of $\omega$ is invariant, moreover one has $\cat{\emph{d}} \eta = \omega$. In particular, $\omega$ is invariant.
\end{lemma}
\begin{proof} Using again the the transformations of Lemma \ref{transCotM}, it is enough to compute 
\bear
 (-1)^{|x_a| + 1} dx_a p_a = (-1)^{|x_a| + 1} dz_c \left ( \frac{\partial x_a}{\partial z_c} \right ) \left ( (-1)^{|x_a| + |z_b|}  \frac{\partial z_b}{\partial x_a} q_b \right ) = (-1)^{|z_b| + 1} dz_b q_b.
\eear
Clearly $\cat{d} \eta = \omega$. Since both $\cat{d}$ and $\eta$ are invariant so is $\omega.$
\end{proof}
\noindent The previous results allow to give the following definition, see for example \cite{Khuda1}, \cite{Severa} or the dedicated chapter in the recent book \cite{Mnev}.
\begin{definition}[Odd Symplectic Form / Odd Symplectic Supermanifold] \label{oddsympl} We call $\omega \defeq \sum_a dx_a dp_a$ the odd symplectic form associated to $\cat{{M}}$. In particular, we say that the pair $(\cat{M}, \omega)$ defines a odd symplectic supermanifold. 
\end{definition}
\begin{remark} Notice that with respect to the definition of odd symplectic supermanifolds available in the literature \cite{Severa}, the supermanifold $\cat{M}$ is constructed by starting from a supermanifold $\mani$ and a vector bundle on it, better than from an ordinary manifold $X$ and a vector bundle on it: in this sense it is a \virgolette generalized'' odd symplectic supermanifold.  
\end{remark}
\noindent Left multiplication by the odd symplectic form $\omega = \sum_a dx_a dp_a$ induces a well-defined invariant operator $\cat{s}: \Omega^{\bullet}_{\tiny{\cat{M}}} \rightarrow \Omega^{\bullet}_{\tiny{\cat{M}}}$ whose action with respect to the above decomposition is given by
\bear
\cat{s}\, (\eta \otimes F \otimes f ) = (-1)^{|x_a| |\eta|} dx_a \eta \otimes dp_a F \otimes f. 
\eear
Such as the de Rham differential $\cat{d}$, also the multiplication by the odd symplectic form $\cat{s} :  \Omega^{\bullet}_{\tiny{\cat{M}}} \rightarrow \Omega^{\bullet}_{\tiny{\cat{M}}}$ is nilpotent. We compute its homology in the next theorem. 
\begin{theorem}[Homology of $\cat{s}$] \label{shomology} Let $\cat{\emph{M}}$ be defined as above and let $\cat{\emph{s}} : \Omega^{\bullet}_{\tiny{\cat{\emph{M}}}} \rightarrow \Omega^{\bullet}_{\tiny{\cat{\emph{M}}}}$ the left multiplication by the odd symplectic form. Then 
\bear
H_{\tiny{\cat{\emph{s}}}} (\Omega^\bullet_{\tiny{\cat{\emph{M}}}}) \cong [dz_1 \ldots dz_n \otimes dp_{n+1} \ldots dp_{n+m}], %\mathcal{B}er (\mani) \otimes_{\mathcal{O}_{\mani}} (\Omega^\bullet_{\mani})^\ast,
\eear
as a sheaf of $\mathcal{O}_{\tiny{\cat{\emph{M}}}} = (\Omega^\bullet_\mani)^\ast$-modules. %In other words, the homology of $\cat{s}$ yields the integral forms on the base supermanifold $\mani.$ 
%In particular, the homology of $\cat{\emph{s}}$ is naturally isomorphic to the sheaf\
%\bear \label{homs}
%H_{\tiny{\cat{\emph{s}}}} (\Omega^\bullet_{\tiny{\cat{\emph{M}}}}) \cong \pi^\ast \mathcal{B}er (\mani) \otimes (\Omega^\bullet_{\mani})^\ast.
%\eear
\end{theorem}
\begin{proof} We need to construct a homotopy for the operator $\cat{s}.$ Using the above local decomposition, $\Omega^\bullet_{\tiny{\cat{M}}}$ can be represented by the sheaf of vector spaces generated by the elements $\eta \otimes F \otimes f$ where $f \in \mathcal{O}_{\tiny{\cat{M}}} \lfloor_{\pi^{-1} (U)}$ and for monomials $\eta = dx^I $ and $F = dp^J$ for multi-indices $I$ and $J $.\\
On the other hand, one can observe that the decomposition \eqref{deco} coincides with $\Omega^\bullet_{\mani} \otimes \mbox{gr} (\mathcal{D}_\mani) \otimes (\Omega^\bullet_{\mani})^\ast$, and in view of this, the action of the operator $\cat{s}$ reads
\bear
\cat{s} (\eta \otimes F \otimes f) = (-1)^{|x_a| |\omega|} dx_a (dx^J) \otimes \partial_{x_a} (\partial^J) \otimes f,
\eear
as $dp^J$ corresponds to $\partial^J$ and $dp_a$ corresponds to $\partial_{x_a}$, having used the local splitting of the $\Omega^1_{\tiny{\cat{M}}}$-extension. Notice that $f$ is not touched by $\cat{s}.$ We thus introduce the following local operator:
\bear
\cat{h} (\eta \otimes F \otimes f) \defeq \sum_a (-1)^{|x_a|( |dx^I| + |\partial^J| + 1)}{\partial_{dx_a}} dx^I \otimes [\partial^J , x_a] \otimes f.
\eear
We prove that this is a homotopy for the operator $\cat{s}.$ In particular, one finds that 
\begin{align} \label{sumhom}
(\cat{h} \cat{s} + \cat{s} \cat{h} &)(\eta \otimes F \otimes f)  = \sum_{a,b} (-1)^{(|x_a| + |x_b|)|\eta| }\delta_{a b} \eta \otimes \partial^J \otimes f + \nonumber \\
& \quad +  \sum_a (-1)^{|\partial^J||x_a|} \eta \otimes \partial_a [\partial^J, x_a] \otimes f + \sum_a (-1)^{|x_a|+ 1} dx_a (\partial_{dx_a} \eta) \otimes \partial^J \otimes f .
\end{align}
The summands in the previous expression read: 
\begin{align}
& \sum_{a,b} (-1)^{(|x_a| + |x_b|) |\eta|} \delta_{ab} \eta \otimes \partial^J \otimes f = (n+m) (\eta \otimes F \otimes f),  \nonumber \\
& \sum_{a} (-1) ^{|x_a||\partial^J|} \eta \otimes \partial_a [\partial^J, x_a] \otimes f = \sum_a (-1)^{|x_a|} \eta \otimes \partial^J \otimes f = ( \mbox{deg}_0 (F) -  \mbox{deg}_1 (F) )(\eta \otimes F \otimes f), \nonumber \\
& \sum_{a} (-1)^{|x_a|+ 1} dx_a (\partial_{dx_a} \eta) \otimes \partial^J \otimes f = ( \mbox{deg}_0 (\eta) -  \mbox{deg}_1 (\eta) )(\eta \otimes F \otimes f).
\end{align}
where $n$ is the even and $m$ is the odd dimension of $\mani$ and $\mbox{deg}_0$ and $\mbox{deg}_1$ is the even and odd degree of $\eta = dx^I$ and $F = \partial^J$. It follows that the above sum \eqref{sumhom} gives 
\bear
(\cat{h} \cat{s} + \cat{s} \cat{h} )(\eta \otimes F \otimes f)  = \left ( (n + m) + (\mbox{deg}_0 (\partial^J) - \mbox{deg}_1 (\partial^J) ) + (\mbox{deg}_0 (\eta) - \mbox{deg}_1 (\eta))  \right ) (\eta \otimes F \otimes f). \nonumber
\eear
The homotopy $\cat{h}$ fails if and only if one has $\mbox{deg}_0 (\eta) = \mbox{deg}_0 (\partial^J) = 0$, $\mbox{deg}_1 (\eta) = n$ and $\mbox{deg}_1 (\partial^J) = m$, so that the non-zero element in homology takes the form $dz_1 \ldots dz_n \otimes \partial_{\theta_1} \ldots \partial_{\theta_m} = dz_1 \ldots dz_n \otimes dp_{n+1} \ldots dp_{n+m} \otimes f$, where $f$ is any section of the structure sheaf $\mathcal{O}_{\tiny{\cat{M}}}$.
\end{proof}
\noindent The above theorem has the following corollary.
\begin{corollary} \label{IntForm} Let $\cat{\emph{M}}$ be defined as above and let $\cat{\emph{s}} : \Omega^{\bullet}_{\tiny{\cat{\emph{M}}}} \rightarrow \Omega^{\bullet}_{\tiny{\cat{\emph{M}}}}$ be the left multiplication by the odd symplectic form. Then the homology of $\cat{\emph{s}}$ is naturally isomorphic to the pull-back of the Berezinian sheaf on $\mani$, \emph{i.e.}\
\bear \label{homs}
H_{\tiny{\cat{\emph{s}}}} (\Omega^\bullet_{\tiny{\cat{\emph{M}}}}) \cong \pi^\ast \mathcal{B}er (\mani).
\eear
\end{corollary}
\begin{proof} Allowing for the above identifications and the usual slight abuse of notation concerning the pull-backs, it is enough to observe that $[dz_1 \ldots dz_n \otimes \partial_{\theta_1} \ldots \partial_{\theta_m}] = [dz_1 \ldots dz_n \otimes dp_{n+1} \ldots dp_{n+m}]$ generates the Berezinian sheaf of the supermanifold $\mani$, see \cite{NojaRe} for details on this construction of the Berezinian sheaf.
\end{proof}
\begin{remark} The above result can be related to the notion of (super) \emph{semidensities}, see \cite{Khuda1}, \cite{Mnev}, \cite{Severa}. Indeed the $\Omega^1_{\tiny{\cat{M}}}$-extension exact sequence \eqref{extomega} allows to easily compute the Berezinian sheaf $\mathcal{B}er (\cat{M}) \defeq \mathcal{B}er (\Omega^1_{\tiny{\cat{M}}})^\ast$ of the supermanifold $\cat{M}$. Taking the Berezinians, the short exact sequence \ref{extomega} yields
\bear
\mathcal{B}er (\Omega^1_{\tiny{\cat{M}}})^\ast \cong \mathcal{B}er (\pi^\ast \Omega^1_{\mani} )^\ast \otimes_{\mathcal{O}_{\cat{\tiny{M}}}} \mathcal{B}er (\pi^\ast \cat{T}_\mani)^\ast \cong \pi^\ast (\mathcal{B}er (\mani) \otimes_{\mathcal{O}_\mani} \mathcal{B}er(\cat{T}_\mani)^\ast).
\eear    
Observing that for any sheaf $\mathcal{E}$ on $\mani$ one has $\mathcal{B}er (\Pi \mathcal{E}) \cong \mathcal{B}er(\mathcal{E})^\ast \cong \mathcal{B}er (\mathcal{E}^\ast)$, one sees that $\mathcal{B}er(\cat{T}_\mani)^\ast \cong \mathcal{B}er (\Pi \cat{T}_\mani^\ast)^\ast = \mathcal{B}er^\ast (\Omega^1_\mani)^\ast = \mathcal{B}er(\mani)$ hence 
\bear \label{semid}
\mathcal{B}er ({\cat{M}}) \cong \pi^\ast \mathcal{B}er (\mani)^{\otimes 2}.
\eear
Defining the sheaf of semidensities $\mathcal{D}\mathpzc{ens} (\cat{M})^{1/2} $ of the supermanifold $\cat{M}$ to be the locally-free sheaf of $\mathcal{O}_{\tiny{\cat{M}}}$-modules whose sections are \virgolette square roots'' of the sections of the Berezinian sheaf, \emph{i.e.} $\mathcal{D}\mathpzc{ens}^{1/2} (\cat{M}) \defeq \mathcal{B}er (\cat{M})^{\otimes 1/2}$, it follows from the \eqref{semid} that $\mathcal{D}\mathpzc{ens}(\cat{M})^{1/2} \cong \pi^\ast \mathcal{B}er (\mani). $
In turn, the above \eqref{homs} can be re-written as 
\bear
H_{\tiny{\cat{{s}}}} (\Omega^\bullet_{\tiny{\cat{{M}}}}) \cong \mathcal{D}\mathpzc{ens} (\cat{M})^{1/2},
\eear 
where the sheaf of semidensities is seens as a sheaf of $\mathcal{O}_{\tiny{\cat{M}}}$-modules. %therefore $\mathcal{D}\mathpzc{ens} (\cat{M})^{1/2} = \mathcal{D}\mathpzc{ens} (\cat{M})^{1/2} \otimes_{(\Omega^\bullet_\mani)^\ast} (\Omega^\bullet_\mani)^\ast.$ 
%This shows that in this context semidensities on $\cat{M}$ are indeed isomorphic to integral forms on $\mani.$ 
Notice that reducing to the underlying ordinary manifold $\mani_{\mathpzc{red}}$ one would find $\mathcal{D}\mathpzc{ens} (\cat{M}_{\mathpzc{red}})^{1/2} \cong \pi^\ast \mathcal{K}_{\mani_{\mathpzc{red}}}$, which is the ordinary notion for semidensities of odd symplectic supermanifolds constructed out of an ordinary manifold $\mani_{\mathpzc{red}}$, see for example \cite{Mnev} and \cite{Severa}.  
\end{remark}
\subsection{Deformed de Rham Complex and BV Laplacian} Now, the crucial observation, originally due to \v{S}evera in \cite{Severa}, is that the nilpotent operators $\cat{d}$ and $\cat{s}$ commutes with each other. This holds true also in the present setting, as the following shows.
\begin{lemma}[$ \cat{d} $ commutes with $\cat{s}$] \label{commut} Let $\cat{\emph{d}}$ and $\cat{\emph{s}}$ be the de Rham differential and the multiplication by the odd symplectic form, then $[\, \cat{\emph{d}}, \cat{\emph{s}} \, ] = 0.$  In particular the triple $(\Omega^\bullet_{\tiny{\cat{\emph{M}}}}, \cat{\emph{s}}, \cat{\emph{d}})$ defines a double complex.
\end{lemma}
\begin{proof} This is a local check. Using the above decomposition, one computes
\begin{align}
{ \cat{d} \circ \cat{s}} ( \eta \otimes F \otimes f ) = (-1)^{|\eta| |x_a| + |x_b| + |x_b||\eta| + |x_b||F|} dx_b dx_a \eta \otimes dp_a F \otimes \partial_{x_b}f + \nonumber \\
+ (-1)^{|\eta| |x_a| + |x_b| + |x_b| |F| + |\eta| + |F| + |x_b||x_a| + 1} dx_a \eta \otimes dp_b dp_a F \otimes \partial_{p_b} F = {\cat{s} \circ \cat{d}} (\eta \otimes F \otimes f),
\end{align}
%On the other hand instead
%\begin{align}
%{\cat{s} \circ \cat{d}} (\omega \otimes F \otimes f) = (-1)^{|\omega| |x_a| + |x_b| + |x_b||\omega| + |x_b||F| +1} dx_b dx_a \omega \otimes dp_a F \otimes \partial_{x_b}f + \nonumber \\
%+ (-1)^{|\omega| |x_a| + |x_b| + |x_b| |F| + |\omega| + |F| + |x_b||x_a|} dx_a \omega \otimes dp_b dp_a F \otimes \partial_{p_b} F,
%\end{align}
which concludes the proof.
 \end{proof}
\noindent It follows from the previous Lemma \ref{commut} that, in particular, $\cat{d}$ acts on the homology of $\cat{s}$. This leads to the following definition.
\begin{definition}[Deformed de Rham Complex / Spectral sequence $E^{\scriptsize{\cat{s}}}_{i}$] We call the double complex $(\Omega^\bullet_{\tiny{\cat{M}}}, \cat{s}, \cat{d})$ the \emph{deformed} de Rham (double) complex of $\cat{M}$. We denote with $E^{\scriptsize{\cat{s}}}_{i}$ the related spectral sequence $(E_i, \delta_i) $ that starts with the differential $\delta_1 = \cat{s}$ and we call it \emph{deformed de Rham spectral sequence}.  
\end{definition}
\noindent Let us now study the deformed de Rham spectral sequence $E^{\scriptsize{\cat{s}}}_i$. 
\begin{theorem}[Semidensities  \& Super {BV} Operator] \label{BVop} Let ${E}^{\scriptsize{\cat{\emph{s}}}}_{i}$ be defined as above. Then 
\begin{enumerate}[leftmargin=*]
\item[(1)] the first page of the spectral sequence ${E}^{\scriptsize{\cat{\emph{s}}}}_{i}$ is isomorphic to semidensities on $\mani$, \emph{i.e.}
\bear
E^{\scriptsize{\cat{\emph{s}}}}_1 \cong \pi^\ast \mathcal{B}er (\mani);
\eear
\item[(2)] the second differential $\delta_2 $ of the spectral sequence ${E}^{\scriptsize{\cat{\emph{s}}}}_{i}$ is zero. In particular the second page of the spectral sequence ${E}^{\scriptsize{\cat{s}}}_{i}$ is given again by
\bear
E^{\scriptsize{\cat{\emph{s}}}}_2 \cong \pi^\ast \mathcal{B}er (\mani);
\eear 
\item[(3)] the third differential $\delta_3$ of the spectral sequence ${E}^{\scriptsize{\cat{\emph{s}}}}_{i}$ is - up to exact terms - the \emph{super BV Laplacian} 
\bear
\xymatrix@R=1.5pt{
\Delta^{\mathpzc{BV}}_2 : \pi^\ast \mathcal{B}er (\mani)  \ar[r] &  \pi^\ast \mathcal{B}er (\mani)  \nonumber \\
\mathcal{D}  f \ar@{|->}[r] &\Delta^{\mathpzc{BV}}_2 ( \mathcal{D}  f) \defeq \mathcal{D} \left ( \sum_a \frac{\partial^2}{\partial {x_a}\partial p_a} f \right )
}
\eear
where $\mathcal{D} = [dx_1 \ldots dx_{n} \otimes dp_{n+1} \ldots dp_{n+m}]$ is a section of $\pi^\ast \mathcal{B}er(\mani)$ and $f  = f (x, p)$  is a section of $\mathcal{O}_{\tiny{\cat{\emph{M}}}} =( \Omega^\bullet_{\mani})^\ast$ In particular, the spectral sequence converges at page three, which is isomorphic to the locally constant sheaf on $\mani$, \emph{i.e.}
\bear
E^{\scriptsize{\cat{\emph{s}}}}_3 \cong \mathbb{K}_{\mani} \cong E^{\scriptsize{\cat{\emph{s}}}}_\infty.
\eear
A representative of this homology class if given by 
\bear \label{repK}
[dx_1 \ldots dx_{n} \otimes dp_{n+1} \ldots dp_{n+m}]  \, x_{n+1} \ldots x_{n+m} p_{1} \ldots p_{n} \in \pi^\ast \mathcal{B}er (\mani).
\eear 
\end{enumerate}
\end{theorem}
\begin{proof} The first point of the Theorem is just Corollary \ref{IntForm}. As for the second point, notice that the corresponding differential is given by the induced action of the de Rham differential $\cat{d}$ on $E^{\scriptsize{\cat{\emph{s}}}}_1$. Referring to Theorem \ref{shomology}, one can observe that the induced action of $\cat{d}$ maps to a zero-homotopic cohomology. More in particular, for immediate use, one can observe that for any $\mathpzc{S} \in E^{\scriptsize{\cat{\emph{s}}}}_1$, one has 
\bear
\cat{d} ( \mathpzc{S}) = \cat{s} ( \mathpzc{T}), \quad \mbox{ with } \quad
\mathpzc{T} \defeq \sum_a \left ( \partial_{dp_a } \partial_{x_a} + \partial_{dx_a} \partial_{p_a} \right ) (\mathpzc{S}),
\eear 
where $\mathcal{S}$ can be taken to be of the form $[dx_1 \ldots dx_n \otimes dp_{n+1} \ldots dp_{n+m}] \otimes f$. \\
The third differential can be easily inferred by noticing that, formally, $\delta_3 = \cat{d} \circ \cat{s}^{-1} \circ \cat{d}$, so that in particular, when acting on an element of $E^{\scriptsize{\cat{\emph{s}}}}_2 = E^{\scriptsize{\cat{\emph{s}}}}_1$ one finds, upon the previous observation 
\bear
\delta_3 (\mathpzc{S}) = \cat{d} (\mathpzc{T}).
\eear
Taking $\mathpzc{S} = [dx_1 \ldots d_n \otimes dp_{n+1} \ldots dp_{n+m}] \, f$ as above it is easy to compute that 
\begin{align}
\cat{d} (\mathpzc{T}) %= \sum_{a,b} (-1)^{(|x_b| + |x_a| + 1)(n+m)} ( dx_a \partial_{dx_b} [dx_1 \ldots dx_n \otimes dp_{n+1} \ldots dp_{n+m}] \partial_{x_a} \partial_{p_b} f ) \nonumber \\
%& \quad + \sum_{a, b} (-1)^{(|x_b| + |x_a| + 1)(n+m)} (dp_a \partial_{dp_b} [dx_1 \ldots dx_n \otimes dp_{n+1} \ldots dp_{n+m}] \partial_{p_a} \partial_{x_b} f ) \nonumber \\
= [dx_1 \ldots dx_n \otimes dp_{n+1} \ldots dp_{n+m}] \left (  \sum_{a} \frac{\partial^2 }{\partial_{x_a} \partial_{p_a}} f (x,p) \right ) + \mbox{exact terms}%in cohomology}.
\end{align}
%which is the (super) BV Laplacian. \\
We now look for a homotopy for this operator. To this end, without loss of generality, we let $f \in \mathcal{O}_{\tiny{\cat{M}}} = (\Omega^\bullet_{\mani})^\ast$ be of the form $f (x,p)\defeq g_I(x) p^I $ for $x_a = x_1 \ldots x_{n} | x_{n+1} \ldots x_{n+m}$ even and odd coordinates of $\mani$ and $I$ a multi-index. We claim that the homotopy for $\Delta_2^{\mathpzc{BV}}$ is given by
\bear \label{homot}
\cat{K} (f) \defeq \sum_a (-1)^{|g|( |x_a|+1) }\left ( \int_0^1 dt \, t^{\ell_f} x_a \cat{P}^\ast_t g_I  \right ) p_{a} p^I, 
\eear
where $t \in [0,1]$, $\cat{P}^\ast_t g (x) = g(tx) $ and $\ell_f$ is a constant, which depends on $f$, that will be fixed later. An attentive computation yields the following 
\begin{align}
(\Delta_2^{\mathpzc{BV}} \circ \cat{K} + \cat{K} \circ \Delta_2^{\mathpzc{BV}} ) (f (x, p)) &  =  f(x,p) - \delta_{\ell_f +1 + \deg_{1} (g_I), 0} \, g_I(0) p^I + \nonumber \\
& + \left (n+m + \deg_{0} (p^I) - \deg_{1} (p^I) - 2\deg_{1} (g_I) - \ell_g - 1 \right ) \int_0^1 dt\, t^{\ell_g} (\cat{P}^{\ast}_t g_I)  p^I. \nonumber
\end{align}
This gives the following condition on $\ell_f$ as to have a homotopy: 
\bear
\ell_f = n+m + \deg_{0} (p^I) - \deg_{1} (p^I) - 2 \deg_{1} (g_I) -1,
\eear
which yields
\bear
(\Delta_2^{\mathpzc{BV}} \circ \cat{K} + \cat{K} \circ \Delta_2^{\mathpzc{BV}})(f(x, p))  = f(x, p) - \delta_{(n+m + \deg_{0} (p^I) - \deg_{1} (p^I) - \deg_{1} (g_I)) , 0} \, g_I(0|\theta) p^I.
\eear
Observing that $\deg_{0} (p^I) \geq 0$, $0\leq \deg_{1 } (p^I) \leq n$ and $0\leq \deg_{1} (g_I) \leq m$, one sees that the homotopy fails only for $\deg_{0} (p^I) = 0$, $\deg_{1} (p^1) = n $ and $\deg_{1} (g_I) = m.$ One thus finds that $f(x, p) = x_{n+1} \ldots x_{n+m} p_{1} \ldots p_{n}$ so that $
k \cdot  [dx_1 \ldots dx_{n} \otimes dp_{n+1} \ldots dp_{n+m}] \otimes x_{n+1} \ldots x_{n+m} p_{1} \ldots p_{n} $ with $k \in \mathbb{R}$ or $k \in \mathbb{C}$ is a representative for $E_3^{\scriptsize{\cat{s}}}$. Finally, it is easy to see that the representative is $\cat{d}$-closed, so that it yields zero when acted by all the higher differentials, concluding the proof.     
\end{proof}
\subsection{Remarks and Outlooks} The above Theorem \ref{BVop} extends to a \virgolette fully'' supergeometric context the beautiful \v{S}evera's result \cite{Severa} for odd symplectic supermanifolds, with possibly the bonus of showing explicitly the homotopy \eqref{homot} of the super BV Laplacian - which is seen here as a morphism of sheaves -, together with the related representative in sheaf cohomology \eqref{repK}. It is to be noted that the form of the homotopy shown above is somewhat general, as the structure of odd nilpotent operators in supergeometry often consists into a \virgolette multiplication'' of an even and an odd part, such as the BV Laplacian above or the de Rham differential - notice indeed that the related complexes of integral forms and of differential forms are quasi-isomorphic, see Theorem \ref{quasi}. Similar structures for homotopies of differentials can be found also in \cite{CNR}  \cite{Ruiperez} \cite{NojaRe} \cite{Noja}.\\
Finally, a remark - or better a warning - about the holomorphic category is in order. Let us consider a generic smooth supermanifold $\mani$ admitting a closed non-degenerate odd 2-form $\omega$, \emph{i.e.} an odd symplectic supermanifold $(\mani, \omega)$. Then, by a well-known result due to Schwarz the supermanifold $(\mani, \omega)$ is globally symplectomorphic to the \virgolette{standard}'' odd symplectic supermanifold constructed as $\cat{M}$ above, starting from the reduced space $\mani_{\mathpzc{red}}$ of $\mani$, and endowed with its standard odd symplectic form $\sum_i dx_i dp_i$, see \cite{Schwarz}. The proof of this fact heavily relies on that every smooth supermanifold is in fact \emph{split}, hence it is itself the total space of a certain vector bundle whose fibers have odd parity. It is then natural to ask what happen in the holomorphic category, where complex supermanifolds can in fact be non-split. In particular, one can ask the following question: does Schwarz's result hold true in the holomorphic category as well? In other words, is it possible to find an example of complex supermanifold admitting a closed non-degenerate odd 2-form which is not globally isomorphic to a supermanifold of the kind of $\cat{M}$ for some manifold $\mani_{\mathpzc{red}}$? Clearly, a non-split complex supermanifold admitting a globally defined odd non-degenerate closed 2-form would provide such a counterexample to Schwarz result. This suggests, in turn, the following question: do the obstruction classes to splitting a complex supermanifold also obstruct the existence of a globally-defined odd non-degenerate closed 2-form? We leave these questions to future works. 

%non-split supermanifolds admitting a non-degenerate odd 2-form exist or do all odd symplectic supermanifold are 

%Getting back to global geometric aspects, while a generalization of Schwarz's result in \cite{Schwarz} proves that any 

%Getting back to global geometric aspects, Schwarz showed in \cite{Schwarz} that any real supermanifold admitting a non-degenerate globally defined odd 2-form, \emph{i.e.}\ a odd symplectic supermanifold \cite{Mnev}, is symplectomorphic to one realized as a total space as $\cat{M}$, but starting from an ordinary manifold $\mani_{\mathpzc{red}}.$ This result 

%by a well-expected generalization of Schwarz's result in \cite{Schwarz} (see also \cite{Mnev}) to our setting, \emph{any} generalized odd symplectic supermanifold is globally isomorphic (or better symplectomorphic) to one of the kind $(\cat{M}, \omega)$

%It is to be noted, as proved, Theorem \ref{BVop} holds true both in the smooth and the holomorphic category. Nonetheless, if it is to be expected that 

\appendix

\section{Extensions of Sheaves} \label{app1}

\noindent For the sake of readability of the paper we recall that an \emph{extension} of sheaves on a manifold $X$ is a short exact sequence of sheaves
\bear \label{ses}
\xymatrix{
0 \ar[r] & \ar[r] \mathcal{A} \ar[r]^i & \mathcal{C} \ar[r]^j & \mathcal{B} \ar[r] & 0.
}
\eear
In particular, we say that $\mathcal{C}$ is an extension of $\mathcal{B}$ by $\mathcal{A}$. It is well-known from homological algebra that extensions are classified up to equivalence %\emph{i.e.}\ up to commutative diagrams  
%\bear
%\xymatrix{
%0 \ar[r] & \ar[r] \ar@{=}[d] \mathcal{A} \ar[r] & \mathcal{C} \ar[r] \ar[d]^{\phi} & \mathcal{B} \ar@{=}[d]\ar[r] & 0\\
%0 \ar[r] & \ar[r] \mathcal{A} \ar[r] & \mathcal{C}^\prime \ar[r] & \mathcal{B} \ar[r] & 0
%}
%\eear
%where $\phi : \mathcal{C} \rightarrow \mathcal{C}^\prime$ is an isomorphism, 
via their cohomology classes $[\xi ] \in Ext^1 (\mathcal{B}, \mathcal{A})$. In particular, we say that an extension is \emph{split} if is equivalent to the trivial extension, \emph{i.e.}\ if $\mathcal{C} \cong \mathcal{A} \oplus \mathcal{B}$ in 
\eqref{ses} 
%\bear
%\xymatrix{
%0 \ar[r] & \ar[r] \mathcal{A} \ar[r] & \mathcal{A} \oplus \mathcal{B} \ar[r] & \mathcal{B} \ar[r] & 0.
%}
%\eear
This is the same as saying that there exists a \emph{retraction} morphism $\pi : \mathcal{B} \rightarrow \mathcal{C} $ splitting the above exact sequence,
%\bear
%\xymatrix{
%0 \ar[r] & \ar[r] \mathcal{A} \ar[r] & \mathcal{C} \ar[r]_{j} & \mathcal{B} \ar@{->}@/_1.2pc/[l]^{\pi} \ar[r] & 0.
%}
%\eear 
\emph{i.e.} $\pi$ has the property that $ j \circ \pi = id_{\mathcal{B}}$. Notice that if $\mathcal{A}$ and $\mathcal{B}$ are locally-free sheaves on $X$, then one has $\mathcal{H}om (\mathcal{B}, \mathcal{A}) \cong \mathcal{A} \otimes \mathcal{B}^\ast$, so that in particular $Ext^1 (\mathcal{B}, \mathcal{A}) \cong H^1 (X, \mathcal{A} \otimes \mathcal{B}^\ast)$. \vspace{.3cm}

%It follows that in order to study the geometry of $\Omega^1_{\tiny{\cat{M}}}$, with reference to equations \eqref{extomega} in Theorem \ref{extteo}, we thus need to consider the cohomology group 
%\bear
%Ext^1 (\pi^\ast \cat{T}_{\mani}, \pi^\ast \Omega^1_{\mani}) \cong H^1 (|\cat{M}_{\mathpzc{red}}|, \mathcal{H}om (\pi^\ast \cat{T}_{\mani}, \pi^\ast \Omega^1_{\mani} )).
%\eear
%We will deal with it in different geometric frameworks in the next subsection.

%\noindent Theorem \ref{splitting} shows that in the smooth category $\Omega^1_{\tiny{\cat{M}}}$ is indeed a split extension of $\pi^\ast \Omega^1_{\mani}$ by $\pi^\ast \cat{T}_\mani$, but the splitting is \emph{non-canonical}.

\noindent We will now compute explicitly the above $Ext$-functor. We will work in a general setting - over a smooth, analytic or algebraic manifold $X$ -, following an ordinary diagram chasing argument. \\
%Working in any geometric category, we start considering an exact sequence of locally-free sheaves on a smooth, analytic of algebraic manifold $X$
%\bear \label{ses}
%\xymatrix{
%0 \ar[r] & \mathcal{A} \ar[r]^{i} &\mathcal{C} \ar[r]^{j} & \mathcal{B}  \ar[r] & 0.
%}
%\eear
Since \eqref{ses} is always locally split, \emph{i.e.}\ $\mathcal{C} \lfloor_{U} \cong \mathcal{A}\lfloor_{U} \oplus \mathcal{B} \lfloor_{U}$ on an open set $U$ in $X$, then there exists a basis $\underline c^U = \{ c^U_1, \ldots, c^U_{n+m}\}$ of $\mathcal{C}$ such that $\underline{a}^U = \{c^U_1, \ldots, c_n^U \}$ is a basis of $\mathcal{A}$ and $\underline{b}^U = \{ j(c^U_{n+1}), \ldots j(c^U_{n+m})\}$ is a basis of $\mathcal{B}$. \\
If now $U$ and $V$ are two open sets in $X$ such that $U \cap V \neq 0$, and $\underline c^U$ and $\underline c^V$ are the related local bases on $U$ and $V$ respectively, then %with reference to the form of the transition functions in Theorem \ref{transCotM}, 
we consider a coordinate transformation of the following form (see for example the transition functions of Theorem \ref{transCotM} in the main text):
\bear \label{transext}
\underline{c}^U = \underline{c}^V \left ( \begin{array}{c|c}
A & C \\
\hline
0 & B
\end{array}
\right ),
\eear
where $A \in \cat{Mat}_{n\times n} (\mathcal{O}_{U \cap V})$, $B \in \cat{Mat}_{m\times m} (\mathcal{O}_{U \cap V})$ and $C \in \cat{Mat}_{n\times m} (\mathcal{O}_{U \cap V}).$ The class $[\xi] \in Ext^1 (\mathcal{B}, \mathcal{A})$ is defined applying the contravariant functor ${H}om ( \cdot , \mathcal{A})$ to the short exact sequence \eqref{ses}, obtaining 
\bear
\xymatrix{
0 \ar[r] & {H}om(\mathcal{B}, \mathcal{A}) \ar[r] & {H}om (\mathcal{C}, \mathcal{A}) \ar[r] & {H}om (\mathcal{A}, \mathcal{A}) \ar[r]^\delta & Ext^1 (\mathcal{B}, \mathcal{A}) \ar[r] & \ldots.
}
\eear
We have $[\xi] = \delta (id_\mathcal{A})$. In order to explicitly compute this, we use a covering having open sets $U$ and $V$ with $U\cap V \neq \emptyset$. In particular, we describe the element $id_\mathcal{A} \in Hom (\mathcal{A}, \mathcal{A}) = \Gamma_X (\mathcal{A} \otimes \mathcal{A}^\ast)$ on the open set $U$ as 
$
\underline a^U \cdot \partial^U_{\underline a} \defeq a^U_i \otimes \partial^U_{a_j},
$
where we have introduced $\{ \partial^U_{a_i}\}_{i=1}^n$, the basis dual to $\{ a_i\}_{i=1}^n$ on $U$. %The above tensor product is usually called Kronecker product. 
Notice that here $\underline a^U$ is looked at as a \emph{row} vector and $\partial^U_{\underline a}$ is looked at as a \emph{column} vector, so that in particular, their transformation in an intersection reads $\underline a^U = \underline a^V A$ and $\partial^U_{\underline a} = A^{-1} \partial^V_{\underline a}$, or analogously $\partial^{U t}_{\underline a} = \partial^{V t}_{\underline a} (A^{-1})^{t}$. It follows that indeed $id_{\mathcal{A}} = \underline{a}^U \cdot \partial^U_{\underline a} = \underline{a}^V A \cdot A^{-1} \partial^V_{\underline a} = \underline{a}^V \cdot \partial^V_{\underline{a}}$. \\
We now look at the transformation of the dual basis $\partial^U_{\underline c}$ of $ \underline{c}^U$, which we decompose - with a slight abuse of notation - as $\partial^U_{\underline c} = (\partial^U_{\underline a}, \partial^U_{\underline b}).$ We have that 
\bear
(\partial^U_{\underline a}, \partial^U_{\underline b}) = (\partial^U_{\underline a}, \partial^U_{\underline b}) \left  ( \begin{array}{c|c}
A & C \\
\hline
0 & B
\end{array}
\right )^{ -1 \,t} = \left  ( \begin{array}{c|c}
(A^{-1})^{t} & 0 \\
\hline
- (B^{-1})^t C^{t} (A^{-1})^t & (B^{-1})^t
\end{array}
\right ).
\eear
In particular, it follows that 
\bear
\partial^{Ut}_{\underline a} = \partial^{Vt}_{\underline a} (A^{-1})^t - \partial^{Vt}_{\underline b} (B^{-1})^t C^t (A^{-1})^t
\eear
Let us now consider the liftings of $id_{\mathcal{A}}$ to $Hom (\mathcal{C}, \mathcal{A})$ and their difference $\varphi_{UV} \defeq \underline a^U \cdot \partial^U_{\underline a} - \underline{a}^V \cdot \partial^V_{\underline a }$ written with respect to the bases on $V$. We find
\bear
\varphi_{UV} = \underline a^U \cdot \partial^U_{\underline a} - \underline{a}^V \cdot \partial^V_{\underline a } = \underline a^V (A \cdot A^{-1}) \partial^U_{\underline a} - \underline{a}^V (A \cdot A^{-1} CB^{-1} ) \partial^V_{\underline b } =  - \underline{a}^V (CB^{-1}) \partial_{\underline b}^V.
\eear
We observe that $\varphi_{UV} $ can be naturally interpreted as a section $\varphi_{UV} \in \Gamma_{U \cap V} ( Hom(\mathcal{B}, \mathcal{A}) )$, whose associated matrix with respect to the bases $\underline a^V$ and $\underline b^V$ is given by $-CB^{-1} \in \cat{Mat}_{n\times m} (\mathcal{O}_{U \cap V}).$ More in general, given a open covering $\mathcal{U} = \{U_i \}_{i\in I}$ of $X$, the cohomology class $[\xi] \in Ext^1 (\mathcal{B}, \mathcal{A}) $ is represented by the cocycle $\{ \varphi_{ij} \}_{i, j \in I}$, such that
\bear
Ext^{1} (\mathcal{B}, \mathcal{A}) \owns [\xi]   \leftrightsquigarrow \{ \varphi_{ij} : U_i\cap U_j \rightarrow \cat{Mat}_{n\times m} (\mathcal{O}_X \lfloor_{U_i \cap U_j}): i <j \},
\eear
which is represented by the matrix $-CB^{-1}$ on in intersection $U_i \cap U_j$, with respect to the bases of $\mathcal{A}$ and $\mathcal{B}$ chosen on $U_j.$\\

\noindent Now, if $[\xi] \equiv 0 $ in cohomology, the related \v{C}ech cocycle is actually a coboundary, \emph{i.e.}\
\bear \label{cob}
\varphi_{UV}  = (\varphi_V - \varphi_{U} ) \lfloor_{U \cap V},
\eear
for two open sets $U\cap V \neq \emptyset$ and where $\varphi_V \in \Gamma_V (Hom (\mathcal{B}, \mathcal{A}))$ and $\varphi_U \in \Gamma_U (Hom (\mathcal{B}, \mathcal{A}))$. In particular, choosing bases on $U$ and $V$, in terms of matrix representatives, one pose $\varphi_U \defeq \underline{a}^U [M_U] \partial^U_{\underline b}$ and $\varphi_V = \underline{a}^V [M_V ] \partial^V_{\underline b}$ and $\varphi_{UV} = \partial_{a}^V [-CB^{-1}] \partial^V_{\underline b}$. Changing coordinates from $U$ to $V$ in $\varphi_U$ one has $\varphi_U = \underline{a}^V [A M_U B^{-1}] \partial^V_{\underline b}$. Substituting these in \eqref{cob} one gets the matrix identity
\bear
-CB^{-1} = M_V - A M_U B^{-1},
\eear
which in turn can be rewritten as 
\bear \label{vanish}
0 = C + M_V B - AM_U.
\eear
Recalling that in (non-abelian) \v{C}ech cohomology, by definition, two 1-cocycles $\{ g_{ij} \}_{i<j}$ and $\{ g^{\prime}_{ij}\}_{i<j}$ are cohomologous if $g_{ij}^\prime = h_i g_{ij} h_j^{-1}$ for some $0$-cochains $\{ h_i \}_{i \in I}$, then in the present case for the sheaf $\mathcal{C}$, it is enough to consider, say on $U_i = V$ 
\bear
h_V = \left (\begin{array}{c|c}
1_n & M_V \\
\hline 
0 & 1_m
\end{array}
\right ),  
\eear
so that, in turn $h_V g_{VU} g_{U}^{-1}$ reads
\bear
\left (
\begin{array}{c|c}
1_n & M_V \\
\hline 
0 & 1_m 
\end{array}
\right ) 
\left (
\begin{array}{c|c}
A & C \\
\hline 
0 & B 
\end{array}
\right )
\left (
\begin{array}{c|c}
1_n & -M_U \\
\hline 
0 & 1_m 
\end{array}
\right )  = 
\left (
\begin{array}{c|c}
A & {C+ M_VB - AM_U} \\
\hline 
0 & B 
\end{array}
\right ) = \left (
\begin{array}{c|c}
A & 0 \\
\hline 
0 & B 
\end{array}
\right ), 
\eear
upon using equation \eqref{vanish} in the last equality. We summarize the above discussion in the following Theorem.
\begin{theorem} \label{reductionlemma} Let $\mathcal{C}$ be an extension of $\mathcal{B}$ by $\mathcal{A}$ as in \eqref{ses} with transition functions of the form 
%\eqref{transext}
\bear \label{transomega}
G = \left ( \begin{array}{c|c}
A & C \\
\hline
0 & B
\end{array}
\right )
\eear
for $A$ and $B$ transition functions of $\mathcal{A}$ and $\mathcal{B}$ respectively, upon choosing local bases. Then $[\xi] \in Ext^1 (\mathcal{B}, \mathcal{A})$ is represented by a 1-cocycle valued in $\cat{\emph{Mat}}_{n\times m} (\mathcal{O}_{X})$ given by $[-CB^{-1}]$ with respect to the chosen bases of $\mathcal{A}$ and $\mathcal{B}$. In particular, if $[\xi] \equiv 0$, then the structure group of the $\mathcal{C}$ reduces to $GL(n) \times GL(m)$.
\end{theorem}

\section{Geometry of the Super Conic} \label{app2}

\noindent For a better understanding of Example \ref{SuperConic}, we briefly spell out in this appendix some details of the geometry of the complex supermanifold $\mathcal{C} \subset \mathbb{CP}^{2|2}$ cut out by the equation  
\bear \label{superc}
X_0^2 + X_1^2 + X_2^2 + \Theta_1 \Theta_2 = 0  \quad \subset \quad \mathbb{CP}^{2|2}.
\eear   
where the $X$'s and the $\Theta$'s denote homogeneous even and odd coordinates of $\mathbb{CP}^{2|2}$ respectively. The above equation \eqref{superc} defines a $1|2$-dimensional (complex) supermanifold. Further it is easy to observe that, setting the odd homogeneous coordinates to zero, one is left with the equation 
\bear
X_0^2 + X_1^2 + X_2^2 = 0  \quad \subset \quad \mathbb{CP}^{2}.
\eear
which defines a conic in $\mathbb{CP}^2$, implying that the reduced manifold of $\mathcal{C}$ is isomorphic to $\mathbb{CP}^1$, that is $\mathcal{C}_{\mathpzc{red}} \cong \mathbb{CP}^1.$ We will show that the super conic $\mathcal{C}$ is indeed isomorphic to a non-projected supermanifold of dimension $1|2$ having $\mathbb{CP}^1$ as reduced manifold. \\
It is not hard to classify all of the non-projected $1|2$-dimensional complex supermanifolds over $\mathbb{CP}^1.$ Indeed the following fundamental Theorem holds true, see Proposition 9 in \cite{Manin}, Chapter 4, Section 2.
\begin{theorem}[Supermanifolds of dimension $n|2$] \label{dim2} Let $\mani \defeq (|\mani_{\mathpzc{red}}|, \mathcal{O}_\mani)$ be a complex supermanifold of dimension $n|2$. Then $\mani$ is defined up to isomorphism by the triple 
$(\mani_{\mathpzc{red}}, \mathcal{F}_\mani, \omega_\mani)$ where $\mathcal{X}_{\mathpzc{red}}$ is the reduced manifold of $\mani$, 
$\mathcal{F}_\mani$ is a locally-free sheaf of $\mathcal{O}_{\mani_{\mathpzc{red}}}$-modules of rank 2 - the fermionic sheaf of $\mani$ -, and $\omega_\mani$ is the fundamental obstruction $\omega_{\mani} \in H^1 (|\mani_{\mathpzc{red}}|, \cat{\emph{T}}_{\mani_{\mathpzc{red}}} \otimes_{\mathcal{O}_{\mani_{\mathpzc{red}}}} \wedge^2 \mathcal{F}_\mani).$
\end{theorem}
\noindent Notice that since the odd dimension is $2$, no higher obstruction classes can appear. 
\begin{remark}
Concretely, in presence of a non-zero obstruction class, the (even) transition functions coming from the underlying manifold $\mani_{\mathpzc{red}}$ get 
a correction coming from $\omega_\mani$ as they are lifted to $\mani$. More precisely, if $\{{U}_i \}_{i \in I}$ is an open covering of $|\mani_{\mathpzc{red}}|$ such that in a certain intersection ${U}_i \cap {U}_j$ the transition functions of $\mani_{\mathpzc{red}}$ 
are given by certain (holomorphic) functions $z_{ij}^{\ell} = z_{ij}^{\ell} (\underline z_j)$ for $\ell = 1, \ldots, n,$ then the \emph{even} transition functions of a non-projected $n|2$-dimensional supermanifold will be given explicitly by 
\begin{equation}\label{eq:ztransf}
z_{ ij}^{\ell} (\underline z_j , \underline \theta_j) = z^{\ell}_{i j} (\underline z_j) + (\omega_{\mani})_{ij} (\underline z_j , \underline \theta_j)(z^\ell_{ ij}) \qquad \ell = 1, \ldots, n,
\end{equation}
where the $z$'s and the $\theta$'s are respectively even and odd local coordinates for $\mani$, and where $ \omega_{ij}$ is a derivation acting on $z_{ij}^\ell$ and taking values in $\wedge^2 \mathcal{F}_\mani$ - hence the $\theta$'s can only appear in $ (\omega_\mani)_{ij}$ through their product, thus respecting parity.
\end{remark}
\noindent Now, keeping fixed $\mani_{\mathpzc{red}} = \mathbb{CP}^1$, Theorem \ref{dim2} yields the following result for non-projected of dimension $1|2$.
\begin{theorem} [Non-Projected $1|2$-dimensional Supermanifolds over $\mathbb{CP}^1$] \label{ManTH} Every non-projected $1|2$-dimensional supermanifold $\mani$ over $\mathbb{CP}^1$ is characterised up to isomorphism by a triple $(\mathbb{CP}^1, \mathcal{F}_\mani, \omega_{\mani})$ where $\mathcal{F}_\mani$ is a locally-free sheaf of $\mathcal{O}_{\mathbb{CP}^1}$-modules of rank 2 such that 
\bear
\mathcal{F}_\mani \cong  \mathcal{O}_{\proj 1} (m) \oplus  \mathcal{O}_{\proj 1} (n),
\eear 
with $m+n =\ell$, 
$\ell \leq 4$ and $\omega_\mani$ is a non-zero cohomology class $\omega_\mani \in H^1 (|\mathbb{CP}^1| , \mathcal{O}_{\mathbb{CP}^1} (2-\ell)).$
\end{theorem} 
\begin{proof} First by Birkhoff-Grothendieck splitting theorem, see \cite{Har}, every locally-free sheaf of $\mathcal{O}_{\mathbb{CP}^1}$-modules of any rank is isomorphic to a direct sum of invertible sheaves, that in turn are all of the form $\mathcal{O}_{\mathbb{CP}^1} (k)$ for some $k \in \mathbb{Z}$ (recall that $ \mbox{Pic} (\mathbb{CP}^1) \cong \mathbb{Z}$), \emph{i.e.}\ if we let $\mathcal{E}$ be a locally-free sheaf of $\mathcal{O}_{\mathbb{CP}^1}$-modules of rank $n$, then we have
$
\mathcal{E} \cong \bigoplus_{i = 1}^n \mathcal{O}_{\proj 1} (k_i)
$
uniquely up to permutation of the terms in direct sum at the right hand side of the isomorphism. It follows that in our case we have $\mathcal{F}_\mani \cong \mathcal{O}_{\mathbb{CP}^1} (m) \oplus \mathcal{O}_{\mathbb{CP}^1} (n)$, for $n, m \in \mathbb{Z}.$ \\
Finally, observing that $\wedge^2 \mathcal{F}_\mani \cong \mathcal{O}_{\mathbb{CP}^1} (n+m)$ and that $\cat{T}_{\mathbb{CP}^1} \cong \mathcal{O}_{\mathbb{CP}^1}(+2)$, one finds that 
\bear \label{cohoCP}
H^1 (|\mathbb{CP}^1| , \cat{T}_{\mathbb{CP}^1} \otimes \wedge^2 \mathcal{F}_\mani) \cong H^1 (|\mathbb{CP}^1| , \mathcal{O}_{\mathbb{CP}^1} (2+m+n))).
\eear
For the supermanifold $\mani$ to be non-projected the cohomology in \eqref{cohoCP} should be non-zero, which amounts to require that $ m+n \leq 4.$ Posing $\ell \defeq m+n \leq 4$ and $\omega_\mani$ a non-zero class in $H^1 (|\mathbb{CP}^1| , \mathcal{O}_{\mathbb{CP}^1} (2+\ell)),$ one concludes using Theorem \ref{ManTH}.
\end{proof}
\noindent We now focus on a particular choice of supermanifold in the \virgolette family'' singled out above, namely we choose $n= m = -2$. Notice that in this case one has a one-dimensional obstruction space, \emph{i.e.}\ $H^1(|\mathbb{CP}^1|, \mathcal{O}_{\mathbb{CP}^1} (-2)) \cong \mathbb{C}$. 
\begin{definition}[The Supermanifold $\mathbb{CP}^{1|2}_{\omega} $] We denote $\mathbb{CP}^{1|2}_\omega $, the supermanifold arising from the triple $(\mathbb{CP}^1, \mathcal{F}, \omega ) $, with $\mathcal{F} =  \mathcal{O}_{\mathbb{CP}^1} (-2)^{\oplus 2}$ and $\omega$ a non-zero class in $H^1(|\mathbb{CP}^1|, \mathcal{O}_{\mathbb{CP}^1} (-2)).$   
\end{definition}
\noindent In order to prove that the supermanifold $\mathbb{CP}^{1|2}_\omega$ is actually isomorphic to the super conic $\mathcal{C}$, we need to find the transition functions of $\mathbb{CP}^{1|2}_\omega$. To this end, let us work with the standard open covering of $\mathbb{CP}^1$ given by $U = (U_0, U_1)$ where $U_i \defeq \{ X_i \neq 0\}$, if $[X_0 : X_1]$ are the homogeneous coordinates of $\mathbb{CP}^1.$ Then one has 
\begin{align} \label{1}
& {U}_0 \defeq \{ X_0 \neq 0 \} \; \rightsquigarrow \; {z} \,\mbox{mod}\,{\mathcal{J}^2} \defeq \frac{X_1}{X_0}, 
%\quad \theta_1 \defeq \frac{\Theta_1}{X_0^{-m}}, \quad \theta_2 \defeq \frac{\Theta_2}{X_0^{-n}} \\
& {U}_1 \defeq \{ X_1 \neq 0 \} \; \rightsquigarrow \; {w}\, \mbox{mod}\,{\mathcal{J}^2} \defeq \frac{X_0}{X_1}, 
%\quad \psi_1 \defeq \frac{\Theta_1}{X_1^{-m}}, \quad \psi_2 \defeq \frac{\Theta_2}{X_1^{-n}}. \label{2}
\end{align}
for $\mathcal{J}$ the nilpotent sheaf of $\mathbb{CP}^{1|2}_\omega$ (note that since we are working in odd dimension $2$ one has $(\mathcal{J})_0 = \mathcal{J}^2 $). Accordingly, on unique intersection $U_0 \cap U_1$ one finds 
\bear \label{reduced}
z\, \mbox{mod}\, \mathcal{J}^2 = \frac{1}{w} \,\mbox{mod}\,{\mathcal{J}^2}. 
%\qquad \theta_1 = \frac{\psi_1}{w^{-m}}, \qquad \theta_2 = \frac{\psi_1}{w^{-n}}.
\eear 
Passing to the fermionic sheaf $\mathcal{F} = \mathcal{O}_{\mathbb{CP}^1} (-2)^{\oplus 2}$, we denote $(\theta_i)_{i = 1, 2}$ a local basis of $\mathcal{F}$ on ${U}_0$ and $(\psi_i)_{i= 1, 2} $ a local basis of $\mathcal{F}$ on ${U}_1$ respectively, so that one can write
\begin{align} 
& {U}_0 \defeq \{ X_0 \neq 0 \} \; \rightsquigarrow \; \theta_i \defeq \pi \left ( \frac{1}{ X_0^{2}}\right ), 
& U_1 \defeq \{ X_1 \neq 0 \} \; \rightsquigarrow \; \psi_1 \defeq \pi \left ( \frac{1}{ X_1^{2}} \right ), \label{2}
\end{align}
where the $\pi$'s are there to remember the odd parity, since $\mathcal{F} = (\mathcal{J})_1$ for $\mathcal{J}$ the sheaf of nilpotent sections in $\mathcal{O}_\mani$. The transition functions in the intersection $U_0 \cap U_1$ are therefore given by
\bear \label{oddt}
\theta_1 = \frac{\psi_1}{w^2}, \qquad \theta_2 = \frac{\psi_1}{w^2},
\eear 
These identify also the product $\theta_1 \theta_2$ (and $\psi_1 \psi_2$) with a section of $\wedge^2 \mathcal{F} \cong \mathcal{O}_{\mathbb{CP}^1} (-4).$ This is enough to give the correction to the even part of the transition functions of the non-projected supermanifold $\mathbb{CP}^{1|2}_\omega$ given by the presence of a non-zero obstruction class $\omega \in H^1(|\mathbb{CP}^1|, \mathcal{O}_{\mathbb{CP}^1} (-2)).$ Indeed, working with the previous conventions, one can explicitly identify 
\bear \label{omegacong}
H^1 (|\mathbb{CP}^1|, \cat{T}_{\mathbb{CP}^1} \otimes \wedge^2 \mathcal{F} )  \cong H^1(|\mathbb{CP}^1|, \mathcal{O}_{\mathbb{CP}^1} (-2)) \owns \omega = {\lambda} \cdot \left [\frac{1}{X_0X_1} \right ], 
\eear
for $\lambda \in \mathbb{C}$ - in our case $\lambda \neq 0.$ Notice that $X_0  X_1 \neq 0$ in $U_1 \cap U_1$, also - with abuse of notation - we will take the liberty of suppressing the index of the intersection as there is a single one of them. The previous \eqref{omegacong} can be rewritten as follows
\bear \label{omegaP}
\omega = \lambda \cdot \left [ \left ( \frac{X_1}{X_0} \right )^3 \frac{1}{X_1^4} X_0^2\right ] = \lambda \cdot \left [ \frac{\psi_1 \psi_2}{w^3} \partial_z \right ],
\eear
where we have identified the sections in the intersection $U_0 \cap U_1$ via $w^3 = (X_0 /X_1)^3$, $\psi_1 \psi_2 = 1 / X^4_1$ and $\partial_z = X^2_0$. Notice indeed that the section of the tangent sheaf $\partial_z \in \cat{T}_{\mathbb{CP}^1}$ satisfies the transformation law $\partial_{z} = -w^2 \partial_w$ and hence it has a double zero at $[0:1] \in \mathbb{CP}^1$, so that it can be identified with the section $X^2_0$ of $\mathcal{O}_{\mathbb{CP}^1} (+2)$. It follows that plugging equation \eqref{reduced} and \eqref{omegaP} into the general expression \eqref{eq:ztransf}, one finds that the even transition functions of $\mathbb{CP}^{1|2}_\omega$ reads
\bear
z (w, \psi_1, \psi_2)   = \frac{1}{w} + \lambda \frac{\psi_1 \psi_2}{w^3},
\eear
for some non-zero complex number $\lambda$, that we will simply set to $1$ in what follows. Indeed, it is true in general that choosing $\omega^\prime_\mani = \lambda \, \omega_\mani$ for some $\lambda \in \mathbb{C}^\ast$ defines an isomorphic extension of $\mathcal{O}_{\mani_{\mathpzc{red}}}$ by $\wedge^2 \mathcal{F}_\mani$ - however the isomorphism is not the identity on $\mathcal{O}_{\mani_{\mathpzc{red}}}$ and $\wedge^2 \mathcal{F}_\mani$. We summarize the previous discussion in the following Lemma. 
\begin{lemma}[Transition Functions of $\mathbb{CP}^{1|2}_\omega$] Let $\mathbb{CP}^{1|2}_\omega$ be the non-projected supermanifold defined as above. Then in the (unique) intersection $U_0 \cap U_1$ the transition functions read
\bear \label{tfCP}
z  = \frac{1}{w} + \frac{\psi_1 \psi_2}{w^3}, \qquad \qquad \theta_i = \frac{\psi_i}{w^2}, \quad i = 1,2.
\eear
\end{lemma}
\noindent We now want to embed $\mathbb{CP}^{1|2}_\omega$ in a projective superspace, namely in $\mathbb{CP}^{2|2}.$ In order to do this, we need to find an ample line bundle $\mathcal{L}$ of $\mathbb{CP}^{1|2}_\omega$ which allows for such an embedding. In the case of $\mathbb{CP}^{1|2}_\omega$ it is easy to define $\mathcal{L}$ using the standard covering of $\mathbb{CP}^{1|2}$ given by $\{U_0, U_1 \}$ introduced above and then giving the expression of the unique transition function in the intersection $U_0 \cap U_1$. Namely, we consider the following 
\bear \label{linebun}
\mathcal{L} \quad \leftrightsquigarrow \quad \left ( \{U_0, U_1\}, e_{U_0} = (w^2 - \psi_1 \psi_2) e_{U_1} \right ),
\eear
where $({e}_{ U_0 }, {e}_{ U_1})$ are the basis or frames of $\mathcal{L}$ on $U_0$ and $U_1$.  
\begin{remark} A comment is in order here. Indeed, as in the ordinary case, one can always describe a line bundle $\mathcal{L}_{\mani}$ on a supermanifold $\mani$ by giving an open covering $\{ U_i \}_{i \in I}$ of $|\mani_{\mathpzc{red}}|$ and the transition functions $\{ g_{ij}\}_{i, j \in I}$ between two local frames $e_{U_i} $ and $e_{U_j}$ in the intersections $U_i \cap U_j$ for $i, j \in I$, so that $e_{U_i} = g_{ij} e_{U_j}$. In this fashion one has indeed the correspondence $\mathcal{L}_{\mani} \leftrightsquigarrow ( \{U_i \}_{i \in I}, \{g_{ij} \}_{i, j \in I})$, where we stress that $g_{ij}$ takes values in $\mathcal{O}^\ast_{\mani, 0} (U_i \cap U_j)$ for any $i,j \in I,$ since the transition functions need to be even, hence parity-preserving. Further, compatility on triple intersections gives a cocycle condition, \emph{i.e.}\ the transition functions $\{ g_{ij}\}_{i, j \in I}$ define classes in $H^1 (|\mani_{\mathpzc{red}}|, \mathcal{O}^\ast_{\mani,0})$. This observation leads to the super-analog of the usual identification of the Picard group $\mbox{Pic} (X)$ of isomorphy classes of line bundle on a complex manifold $X$ with $H^1 (|X|, \mathcal{O}^\ast_X)$: in the case of a supermanifold $\mani$ we have instead $\mbox{Pic} (\mani) \cong H^1 (|\mani_{\mathpzc{red}}|, \mathcal{O}^\ast_{\mani, 0}).$ \\
Along this line of thought, the explicit form of the transition functions \eqref{tfCP} of $\mathbb{CP}^{1|2}_\omega$ comes in handy to verify that the previous definition \eqref{linebun} of $\mathcal{L}$ is well-posed. Indeed, one can check that the transition function defines an element in the cohomology group $H^1 (|\mathbb{CP}^1|, \mathcal{O}^\ast_{\mathbb{CP}^{1|2}_\omega, 0} )$, which is identified with the Picard group $\mbox{Pic}(\mathbb{CP}^{1|2}_\omega)$ of the supermanifold $\mathbb{CP}^{1|2}_\omega$. More in general, the transition functions \eqref{tfCP} allow to compute, via \v{C}ech cohomology, the full Picard group of $\mathbb{CP}^{1|2}_\omega$. Namely, one finds that the Picard group of $\mathbb{CP}^{1|2}_\omega$ is made of lifts of line bundles on $\mathbb{CP}^{1}$ - recall that $\mbox{Pic} (\mathbb{CP}^1) \cong \mathbb{Z}$ - and a continuous part. Namely, one finds  
$
\mbox{Pic} (\mathbb{CP}^{1|2}_\omega) \cong \mathbb{Z}\oplus \mathbb{C}^{3}.
$

\end{remark}

%Indeed, one can check that the transition function defines an element in the cohomology group $H^1 (|\mathbb{CP}^1|, \mathcal{O}^\ast_{\mathbb{CP}^{1|2}_\omega, 0} )$, which is identified in the usual fashion (via transition functions) with the Picard group $\mbox{Pic}(\mathbb{CP}^{1|2}_\omega)$ of the supermanifold $\mathbb{CP}^{1|2}_\omega$. \\
\noindent Getting back to the line bundle $\mathcal{L}$ defined in \eqref{linebun}, the transition functions allow to verify that the following are global sections:
\begin{align}
X_0 \defeq \{ e_{U_0} ,  \; (w^2 - \psi_1 \psi_2) e_{U_1}\}, \quad X_1 \defeq \{ z e_{U_0} , \; w e_{U_1}\}, \quad X_2 \defeq \{ (z^2 -  \theta_1 \theta_2 ) e_{U_0} ,\; e_{U_1}\},
\end{align}
\begin{align}
\Theta_1 \defeq \{ \theta_1 e_{U_0} ,\;  \psi_1 e_{U_1} \}, \quad \Theta_2 \defeq \{ \theta_2 e_{U_0} , \; \psi_2 e_{U_1} \}.
\end{align}
Let us check, for instance, that $X_2$ is a global section:
\begin{align}
%Y_2 & = 
\left (z^2 - \theta_1 \theta_2 \right ) e_{U_0}  
%& = \left (z^2 -  \theta_1 \theta_2 \right ) \left ( w^2 -  \psi_1 \psi_2 \right ) e_{\mathcal{V}} \nonumber \\
 &= \left ( \left [ \frac{1}{w} + \frac{\psi_1 \psi_2}{w^3}\right ]^2 + \frac{\psi_1 \psi_2}{w^4} \right ) \left ( w^2 - \psi_1 \psi_2 \right ) e_{U_1} \nonumber \\
%& = \left ( \frac{1}{w^2} + 2 \frac{\psi_1 \psi_2}{w^4}+  \frac{\psi_1 \psi_2}{w^{4}} \right )\left ( w^2 - \psi_1 \psi_2 \right ) e_{U_1} \nonumber \\
 &= \left ( \frac{1}{w^2} + \frac{\psi_1 \psi_2}{w^4} \right )\left ( w^2 - \psi_1 \psi_2 \right ) e_{U_1}   = e_{U_1}. 
\end{align}
It is immediate to find the equation satisfied by these global sections using their local definitions. Working on $U_0$, for example one has 
\begin{align}
& \left [\Theta_1 \Theta_2 - X^2_1 -X_0 X_2 \right ] |_{U_0} = \theta_1 \theta_2 - z^2 +z^2 -  \theta_1 \theta_2 = 0,
\end{align}
and we leave to the reader to write down the corresponding map $\varphi : \mathbb{CP}^{1|2}_\omega \rightarrow \mathbb{CP}^{2|2}$ and check that it defines an embedding whose image is given by the equation
\bear \label{Cform1}
\Theta_1 \Theta_2 - X^2_1 -X_0 X_2 = 0 \quad \subset \quad \mathbb{CP}^{2|2}.
\eear 
Finally, in order to conclude the verification that $\mathbb{CP}^{1|2}_\omega$ is actually isomorphic to the supermanifold $\mathcal{C} \subset \mathbb{CP}^{2|2}$, one can bring the equation \eqref{Cform1} in the form \eqref{superc} via a transformation in $PGL(3|2)$ - the supergroup of automorphisms of $\mathbb{CP}^{2|2}$. Namely, this is achieved by the transformation 
\bear
PGL (3|2) \owns [T] = \left ( \begin{array}{ccc|cc} 
1 & 0 & i & 0 & 0 \\
0 & i & 0 & 0 & 0 \\
1 & 0 & -i & 0 & 0 \\
\hline
0 & 0 & 0 & 1 & 0 \\
0 & 0 & 0 & 0 & 1 
\end{array}
\right ).
\eear
\noindent We summarize the previous - rather informal - discussion in the following Lemma. 
\begin{lemma}[Super Conic] The complex supermanifold $\mathcal{C} \subset \mathbb{CP}^{2|2}$ cut out by the equation 
\bear
X_0^2 + X_1^2 + X_2^2 + \Theta_1 \Theta_2 = 0  \quad \subset \quad \mathbb{CP}^{2|2},
\eear
is isomorphic to the $1|2$-dimensional supermanifold $\mathbb{CP}^{1|2}_\omega$, determined (up to isomorphism) by the triple $(\mathbb{CP}^{1}, \mathcal{O}_{\mathbb{CP}^1} (-2)^{\oplus 2}, \omega)$, where $\omega$ is a non-zero cohomology class in $H^1 (|\mathbb{CP}^1|, \cat{\emph{T}}_{\mathbb{CP}^1} (-4)).$
\end{lemma}

\section*{Data Availability}

\noindent Data sharing not applicable to this article as no datasets were generated or analyzed during the current study.

\end{document}